\documentclass[10pt]{IEEEtran} 

\usepackage[dvips]{graphicx}  
\usepackage[latin1]{inputenc} 
\usepackage{amssymb}          
\usepackage{amsmath}          
\usepackage{amsfonts}
\usepackage[usenames,dvipsnames]{xcolor}
\usepackage{curves}
\usepackage{comment}

\newtheorem{cor}{Corollary}
\newtheorem{lemma}{Lemma}

\newtheorem{prop}{Proposition}
\newtheorem{defin}{Definition}
\newtheorem{thm}{Theorem}

\newcommand{\ket}[1]{|#1\rangle}
\newcommand{\bra}[1]{\langle #1|}
\newcommand{\braket}[2]{\langle #1|#2\rangle}

\newcommand{\Hi}{\mathcal{H}}

\newcommand{\Ei}{\mathcal{E}}
\newcommand{\trace}{\mathrm{{Tr}}}

\newcommand{\q}[1]{\vert #1 \rangle}
\newcommand{\qd}[1]{\langle #1 \vert}
\newcommand{\compl}[1]{\overline{#1}}
\newcommand{\Cons}{\mathcal{C}}
\newcommand{\Perm}{\mathfrak{P}}

\newcommand{\hi}{\mathcal{H}}
\newcommand{\C}{\mathbb{C}}
\newcommand{\PP}{\mathbb{P}}

\newcommand{\R}{\mathbb{R}}

\newcommand{\diag}{\textrm{diag}}

\newcommand{\qed}{\hfill $\Box$ \vskip 2ex}

 \newcommand{\beq}{\begin{equation}}
 \newcommand{\eeq}{\end{equation}}
 \newcommand{\beqa}{\begin{eqnarray}}
 \newcommand{\eeqa}{\end{eqnarray}}
 \newcommand{\beqan}{\begin{eqnarray*}}
 \newcommand{\eeqan}{\end{eqnarray*}}
 \newcommand{\bea}{\begin{eqnarray}}
 \newcommand{\eea}{\end{eqnarray}}

\title{Consensus for Quantum Networks:\\ From Symmetry to Gossip Iterations}

\author{Luca Mazzarella\thanks{L. Mazzarella is with Dipartimento di Ingegneria
dell'Informazione, Universit\`a di Padova, via Gradenigo 6/B,
35131 Padova, Italy ({\tt mazzarella@dei.unipd.it}).}, Alain Sarlette\thanks{A. Sarlette is with SYSTeMS, Ghent University, Technologiepark Zwijnaarde 914, 9052 Gent, Belgium ({\tt alain.sarlette@ugent.be}).}, Francesco Ticozzi\thanks{F. Ticozzi is with Dipartimento di Ingegneria
dell'Informazione, Universit\`a di Padova, via Gradenigo 6/B,
35131 Padova, Italy ({\tt ticozzi@dei.unipd.it}) and Dept. of Physics and Astronomy, Dartmouth College, 6127 Wilder, 03755 Hanover, NH (USA).} \thanks{Work partially supported by the University of Padua QFUTURE grant. }}

\date{\today}

\begin{document}
\maketitle

\begin{abstract}
This paper extends the consensus framework, widely studied in the literature on distributed computing and control algorithms, to networks of quantum systems. We define consensus situations on the basis of invariance and symmetry properties, finding four different probabilistic generalizations of classical consensus states. We then extend the gossip consensus algorithm to the quantum setting and prove its convergence properties, showing how it converges to symmetric states while preserving the expectation of permutation-invariant global observables.
\end{abstract}



\section{Introduction}\label{sec:intro}

Among the recent trends in control and systems theory, the field of distributed control, estimation and optimization on networks has stimulated an impressive amount of research, see e.g.~\cite{Gorinevsky2008,Bamieh2002a,D'Andrea2003,BulloBook,BarHesp1,VaraiyaEst82,BoydADMM,leonard2007collective,ZampieriGen}. A basic task for distributed information processing is reaching \emph{consensus} on some common control objective, shared value or slack variable. The present paper extends this well-studied consensus problem (see e.g.~\cite{TsitsiklisThesis,Jadbabaie,ConsensusReview,Moreau2005}) to networks of quantum systems, a special case of which would be classical probability distributions.

Exploring the links between information processing tasks and stochastic dynamics on networks has recently opened new research directions towards ``distributed'' quantum information applications. Among these, we recall quantum computation \cite{nielsen-chuang,kitaev-book} in its potential implementation via dissipative means \cite{verstraete2009}, and its connection to quantum random walks \cite{kempe-randomwalks,szegedy-randomwalks}. Other specific tasks include noise protection and dynamical error-correction \cite{viola-IPSlong,ticozzi-isometries,ticozzi-QDS,ahn-feedback,ticozzi-feedbackDD}, entanglement generation through stabilizing dissipative dynamics \cite{kraus-entanglement, ticozzi-QL}, the creation of quantum ``gadgets'' driven by dissipation \cite{wolf-gadgets}, as well as most tasks in the control of open quantum systems \cite{altafini-tutorial,dalessandro-book,james-passivity,gough-product,bolognani-arxiv}.

In this spirit, we here develop a framework for addressing quantum consensus problems: we identify and characterize a hierarchy of quantum consensus situations and study how these can be reached by  suitable dissipative quantum dynamics, while preserving some global information on the network. Our work offers not only a formal generalization of the well-known classical consensus problem, along with some intriguing insights on the latter, but also a potentially new viewpoint on a number of issues in quantum information. More specifically, it ties the structure of symmetric states and correlations \cite{werner-states,werner-symm} to their potential generation via quasi-local resources \cite{ticozzi-QL}, and to the relaxation of many-body systems to thermal states. While this work is mainly meant to provide solid theoretical foundation to quantum consensus problems, links between the proposed theory and potential applications will be outlined in the conclusions. As an interesting by-product, the present work also provides a viewpoint on consensus for classical systems in the context of probability distributions instead of deterministic states.

An attempt to lift the consensus problem to the quantum domain has been presented in \cite{SSRconsensus}.
It is based on a ``cone geometry'' approach, viewing quantum Kraus maps as the non-commutative generalization of Markov chains that model consensus algorithms. The authors show how Birkhoff's Theorem and Hilbert's projective metric lead to a general convergence result and contraction ratio estimation. However, by describing the dynamics of the whole system of interest as governed by a single Markov transition mechanism, this formulation does not account for subsystem structure or network connections in the quantum setting. It therefore defines consensus as asymptotic convergence to a scalar multiple of the identity: for quantum states this corresponds to a fully mixed, most uncertain state which is not the typical desired target for quantum information processing applications.

In the present paper, we approach quantum consensus from an ``operational'', multi-agent control perspective, starting from the basic classical ingredients: a network of subsystems, an interaction protocol with locality constraints, and a target consensus situation.
As a first step, Section \ref{sec:Definition} provides four possible ways to generalize the concept of a \emph{consensus state} to the quantum domain and explores their connections, establishing a hierarchy of quantum consensus definitions. Section \ref{sec:model} presents the quantum open-system dynamics and the locality notions that we employ to describe the interactions between the quantum ``agents''. Typical methods for interaction selection and timing are also introduced, in analogy with classical consensus. This provides all the ingredients needed to specify the general properties of an evolution that achieves \emph{quantum average consensus}. The symmetry-based reformulations at each step provide an alternative interpretation of classical consensus, and in Section \ref{sec:Algorithms} this leads us to a quantum generalization of the classical gossip algorithm. We prove its \emph{convergence to symmetric-state consensus} while preserving the expectation of any permutation-invariant observable. We further show how the algorithm can be explicitly seen as a generalization of the classical one and the classical convergence results can be used to prove a weaker convergence property. The section concludes with an example of the gossip algorithm working on a four-qubit network. In Section \ref{sec:conclusions} we summarize the main results and provide an outlook on possible developments and applications for our framework. A small tutorial on quantum systems modeling is given in Appendix \ref{app:notations}, along with an overview of the notations and conventions we use.


\section{Consensus States}\label{sec:Definition}

A classical \emph{consensus state} for a multipartite system is one in which the states of all the subsystems, often called agents, are the same (although not necessarily stationary). Since in general a quantum state for a multipartite system cannot be factorized into subsystem states (see Appendix \ref{miltipartite}), the definition of consensus must be carefully reconsidered. We here present an operational approach that leads to a hierarchy of definitions for what can be claimed to be ``quantum consensus''.


\subsection{Defining Classical Consensus}\label{subsec:def:classical}

The consensus problem for classical systems is most typically formulated along the following lines~\cite{TsitsiklisThesis,ConsensusReview}. Consider $m$ subsystems, each    one associated to a state given by a configuration variable $x_k \in \R^n$. These subsystems evolve through bilateral interactions, according to some networking scheme and dynamics (see Section \ref{sec:Algorithms}), and reach a \emph{consensus state} if they converge to the set $\Cons = \{\, (x_1,x_2,...,x_m) : x_j = x_k \, \forall \, j,k \,\}$. Furthermore, the agents are said to compute an \emph{average consensus} if for given initial states $(x_1(0),x_2(0),...,x_m(0))$ they converge to the particular equilibrium $(\bar{x},\bar{x},...,\bar{x}) \in \Cons$ where $\bar{x} = \tfrac{1}{m} \sum_i \, x_i(0)$.

Alternatively, consensus can be characterized as invariance w.r.t.~subsystem permutations. Let $\Perm$ denote the set of all subsystem permutation operators, i.e.~each $P_{\pi} \in \Perm$ is associated to some permutation $\pi$ of the integers $1,2,...m$ such that $P_{\pi}\, (x_1,x_2,...,x_m) = (x_{\pi(1)},x_{\pi(2)},...,x_{\pi(m)})$ for any $x_1,x_2,...,x_m$. Denoting the joint state of the $x_k \in \R^n$ by a vector $x \in \R^{mn}$, each $P_{\pi} \in \Perm$ can be written as an $mn \times mn$ matrix resulting from the Kronecker product of some $m \times m$ permutation matrix with the $n \times n$ identity matrix. Then we can define consensus as:
\begin{equation}\label{eq:defconsclass}
	\mathcal{C} = \{ x \in \R^{mn} :\;  P_{\pi}\, x = x \; \text{ for all } P_{\pi} \in \Perm \} \, .
\end{equation}
It is standard that checking $P_{\pi}\, x = x$ for all pairwise permutations $P_{\pi}$ is sufficient to guarantee $P_{\pi}\, x = x$ for all $P_{\pi} \in \Perm$. 

Average consensus is a fundamental element of distributed computation because for each linear function $Q : \R^{mn} \rightarrow \R^s$ which is invariant under all subsystem permutations, there exists $q : \R^n \rightarrow \R^s$ such that $Q x = q \bar{x}$. Thus such $Q$ can be evaluated locally by each agent once they have reached average consensus. The fact that all permutations can be obtained by concatenating pairwise permutations, already suggests that average consensus is computable with standard distributed algorithms.

 
\subsection{Quantum Consensus Definitions and their Relationships}\label{subsec:def:quantum}

Defining what a consensus situation ought to be in a quantum ``network'' is not a straightforward task. More than one definition may be appropriate depending on the type of {\em symmetry} we are seeking. Following the analogy with the classical case can help, but quantum measurement outcomes are intrinsically stochastic, so we must consider {\em probabilistic} consensus situations from the beginning. Let us explore different options by first discussing a simple case.\\

\noindent\textbf{Example 1: When is a quantum network in consensus?} Consider a multipartite quantum system composed of three qubits, with associated Hilbert space $\Hi=\C^2\otimes\C^2\otimes\C^2$, and three observables of the form $\sigma^{(1)}=\sigma_z\otimes I \otimes I,$ $\sigma^{(2)}=I\otimes \sigma_z \otimes I,$ $\sigma^{(3)}=I\otimes I \otimes \sigma_z,$ where the Pauli matrix $\sigma_z=\diag(1,-1).$\footnote{A qubit is a quantum system associated to a two-dimensional Hilbert space $\Hi\sim\C^2$; a standard basis for the latter is conventionally given by the vectors $\{\ket{0},\ket{1}\} \subset \Hi.$ The traceless unitary and hermitian Pauli operators $\sigma_x,\sigma_y,\sigma_z,$ completed with the identity operator, form an orthonormal basis for the operators on $\Hi.$ Explicitly, $\sigma_x=\q{1}\qd{0} + \q{0}\qd{1}, \sigma_y=i\q{1}\qd{0} -i\q{0}\qd{1}, \sigma_z=\q{0}\qd{0} - \q{1}\qd{1}$.} These correspond to observables of the quantity associated to $\sigma_z$ for each of the subsystems. It seems natural to say that the system is in consensus with respect to the expectation of $\sigma_z$ if 
\beq\trace(\rho\sigma^{(1)})=\trace(\rho\sigma^{(2)})=\trace(\rho\sigma^{(3)}).\label{exp}\eeq
The conditions for this to happen can be worked out explicitly in terms of the diagonal elements of the state $\rho$. In particular it is easy to check that all the following states satisfy condition \eqref{exp}:
\beqan
\rho^A &=& \frac{1}{8} I \otimes (\ket{0}+\ket{1})(\bra{0}+\bra{1})\otimes (\ket{0}+\ket{1})(\bra{0}+\bra{1});\\
\rho^B &=& \frac{1}{4} I \otimes (\ket{0,0}+\ket{1,1})(\bra{0,0}+\bra{1,1});\\
\rho^C &=& \frac{1}{{8}} I \otimes I \otimes I;\\
\rho^D &=& \frac{1}{{2}}(\ket{0,0,0}\bra{0,0,0}+\ket{1,1,1}\bra{1,1,1});\\
\rho^E &=& \ket{0,0,0}\bra{0,0,0};\\
\rho^F &=& \frac{1}{{2}}(\ket{0,0,0}+\ket{1,1,1})(\bra{0,0,0}+\bra{1,1,1})\, .
\eeqan
All these states, except $\rho^E$, have $\trace(\rho\sigma^{(i)}) = 0$ for $i=1,2,3$.

The above requirement can be strengthened by requesting \eqref{exp} to hold when $\sigma_z$ is replaced by \emph{any} observable $\sigma \in {\mathfrak B}(\C^2)$ in the definition of $\sigma^{(1)},\sigma^{(2)},\sigma^{(3)}$. This is equivalent to imposing that the reduced states for the three subsystems are the same (a formal proof is provided in the next subsection). It is then easy to check that $\rho^{B},\rho^{C},\rho^{D},\rho^{E},\rho^F$ satisfy this requirement, while $\rho^A$ does not. In fact the reduced states for $\rho^A$ are:
\[\rho^A_1=\frac{1}{2}I,\quad \rho^A_{2} = \rho^A_{3}=\frac{1}{2}(\ket{0}+\ket{1})(\bra{0}+\bra{1}).\]

In the light of \eqref{eq:defconsclass}, another potential definition of quantum consensus would require the {\em state to be symmetric}, i.e.~invariant under any permutations of the subsystems. This choice can be motivated by the classical case, where the consensus state is indeed permutation invariant. Among the states defined in Example 1, only $\rho^{C},\rho^{D},\rho^{E},\rho^F$ are permutation invariant.

Lastly, one might desire to have subsystem agreement not only on the observable averages, but also {\em on each measurement}. Namely, we want that each projective measurement of the (commuting and hence compatible) observables $\sigma_{1},\sigma_{2},\sigma_{3}$ gives \emph{perfectly correlated results} for the three subsystems. That is, among all possible measurement results $\{-1,+1\}^{\times 3}$, only $(-1,-1,-1)$ and $(+1,+1,+1)$ have a nonzero probability to occur\footnote{The set $\{c_1,c_2,c_3,...\}^{\times n}$ is the cartesian product of $\{c_1,c_2,c_3,...\}$  by itself $n$ times, i.e.~the set of $n$-tuples with components taken from $\{c_1,c_2,c_3,...\}$.}. For instance if the system is in an entangled state of the form $\rho = (\alpha\ket{0,0,0}+\beta\ket{1,1,1})(\alpha^*\bra{0,0,0}+\beta^*\bra{1,1,1}),$ with $|\alpha|^2+|\beta|^2=1$, the outcome of a measurement operation is $(-1,-1,-1)$ with probability $|\alpha|^2$ and $(+1,+1,+1)$ with probability $|\beta|^2$. With $\alpha=\beta=\frac{1}{\sqrt{2}}$ we get $\rho^F$ above. The states $\rho^A$, $\rho^B$ and $\rho^C$ do not satisfy this definition of consensus; indeed, for these three states, the distribution of measurement results for qubit 1 is independent of the measurements on the other two qubits, e.g.~$(+1,-1,-1)$ and $(-1,-1,-1)$ have the same nonzero probability so uncorrelated measurement results can happen.
Note however that also mixed states can lead to correlated results, if they express perfect classical correlations, as is the case for any state of the form $\rho^G=p\q{0,0\ldots 0}\qd{0,0\ldots 0}+(1-p)\q{1,1\ldots 1}\qd{1,1\ldots 1}$ with $0<p<1$. For three qubits with $p=\frac{1}{2}$ we get $\rho^D$ above.\hfill $\square$\\

Let us formalize the ideas emerging from the former example. Consider a multipartite system composed of $m$ isomorphic subsystems, labeled with indices $i=1,\ldots,m,$ with associated Hilbert space  $\Hi^m:=\Hi_1\otimes \dots \otimes \Hi_m\simeq \Hi^{\otimes m}$, with $\dim(\Hi_i)=\dim(\Hi)=n$ and $n\geqslant 2$. We shall refer to this multipartite system as to our {\em quantum network}. 
For any operator $X \in \mathfrak{B}(\mathcal{H})$, we will denote by $X^{\otimes m}$ the tensor product $X \otimes X \otimes ... \otimes X$ with $m$ factors. Given an operator $\sigma\in{\mathfrak B}(\Hi)$, we denote by $\sigma^{(i)}$ the local operator:
\[\sigma^{(i)}:=I^{\otimes (i-1)}\otimes \sigma \otimes I^{\otimes (m-i)}.\]
Permutations of quantum subsystems are expressed by a unitary operator $U_{\pi} \in {\mathfrak U}(\Hi)$, which is uniquely defined by 
\[U_{\pi} (X_1\otimes\ldots \otimes X_m) U_{\pi}^\dag= X_{\pi(1)}\otimes\ldots \otimes X_{\pi(m)}\]
for any operators $X_1,\ldots X_m$ in ${\mathfrak B}(\Hi)$, where $\pi$ is a permutation of the first $m$ integers. A state or observable is said to be \emph{permutation invariant} if it commutes with all the subsystem permutations. It is worth noting that given any observable $Q\in \mathfrak{H}(\hi^m)$, we can define a permutation invariant observable $X$ by considering:
\begin{equation}
\label{permutinvL}
X = \frac{1}{m!} \, \sum_{\pi \in \mathfrak{P}} U_{\pi}^{\dagger}QU_{\pi} \, .
\end{equation}

\begin{defin}[$\sigma$EC]\label{def:sigmaEC} Given $\sigma\in{\mathfrak B}(\Hi)$, a state $\rho\in{\mathfrak D}(\Hi^m)$ is in {\em $\sigma$-Expectation Consensus} ($\sigma$EC) if: 
\[\trace(\sigma^{(1)}\rho)=\ldots =\trace(\sigma^{(k)}\rho).\]
\end{defin}

\begin{defin}[RSC] A state $\rho\in{\mathfrak D}(\Hi^m)$ is in {\em Reduced State Consensus} (RSC) if, defining the reduced states $\bar\rho_k=\trace_{\left(\bigotimes_{j\neq k}\Hi_j\right)}(\rho)$, we have 
\[\bar\rho_1=\ldots =\bar\rho_m.\]
\end{defin}

\begin{defin}[SSC] A state $\rho\in{\mathfrak D}(\Hi^m)$ is in {\em Symmetric State Consensus} (SSC) if for each unitary permutation $U_{\pi}$ we have 
\[U_{\pi}\, \rho\, U_{\pi}^\dag = \rho\]
\end{defin}

\begin{defin}[$\sigma$SMC] Given an observable $\sigma$ with spectral decomposition $\sigma=\sum_j s_j \Pi_j \, \in{\mathfrak D}(\Hi)$, all $s_j$ different, a state $\rho\in{\mathfrak D}(\Hi^m)$ is in {\em Single $\sigma$-Measurement Consensus} ($\sigma$SMC) if:
\beq \trace(\Pi_j^{(k)}\Pi_j^{(\ell)}\rho)=\trace(\Pi_j^{(\ell)}\rho),\label{SMCeq}\eeq
for all $k,\ell\in\{1,\ldots,m\},$ and for each $j$.
\end{defin}
The definition of $\sigma$SMC requires that the outcomes of $\sigma$ measurements on different subsystems be exactly the same {\em for each trial}. Indeed, in this last definition, the right-hand side of \eqref{SMCeq} is the probability of obtaining  $s_j$ as a measurement result on both subsystem $\ell$ and $k$ (note that $\Pi_j^{(k)}$ and $\Pi_j^{(\ell)}$ commute, so this joint measurement $\Pi_j^{(k)}\Pi_j^{(\ell)}$ is well-defined). If those two probabilities are equal, then necessarily the probability of  $s_j$ on $k$ conditional to observing $s_j$ on $\ell$ is one (assuming that $\Pi_j^{(\ell)}\rho\Pi_j^{(\ell)} \neq 0$; that special case is trivial and can be treated separately).

One may wonder why, as we did for the consensus in expectation, we did not try to strengthen the $\sigma$SMC property, requiring perfect correlations of outcomes for {\em any} local measurements. In fact it can be shown that states satisfying this property do not exist.
\begin{prop}\label{prop:Nogo}
For any finite-dimensional quantum system, there does not exist a state $\rho$ that satisfies $\sigma$SMC for all $\sigma$.
\end{prop}
The proof makes use of Proposition \ref{propSYM} below, thus we postpone it to Appendix \ref{nosmc}.\hfill $\square$

{\em Remark:} It is worth remarking how all these definitions could be given for classical systems, in the context of consensus for random variables or for probability distributions of the state values. In this case, for example, $\sigma$EC would require the expectation of a set of random variables, each one associated to a subsystem, to be the same in all subsystems; RSC would require the marginal distributions on each subsystem to be equal; and SSC would require that the joint probability distribution is invariant with respect to subsystem permutations. In some sense the definition of $\sigma$SMC is the closest to the classical case, as it requires perfect agreement on the outcome of a set of random variables for {\em each measurement}.

All the states in our example satisfy $\sigma_z$EC, all but $\rho^A$ satisfy RSC, $\rho^C$ to $\rho^F$ satisfy SSC, and $\rho^D$ to $\rho^F$ satisfy $\sigma_z$SMC. There obviously seems to be a hierarchy in these definitions, and the following properties are meant to better characterize them. We shall use the notation ${\cal C}_{\rm XYZ}\subset {\mathfrak D}(\Hi)$ to indicate the set of states that satisfy the properties of the acronym XYZ, e.g.~${\cal C}_{\rm RSC}$ is the set of all Reduced State Consensus states. Note that, since all consensus definitions involve linear constraints, the consensus sets are all convex. 

\begin{thm}\label{hierarchy}
The following chain of implications holds:
\[\mbox{SSC} \implies \mbox{RSC} \implies \mbox{$\sigma$EC},\]
while the converse implications are not true in general.
\end{thm}
\proof {\em SSC $\implies$ RSC}: If $U_{\pi}\rho U_{\pi}^\dag = \rho$ for each permutation, consider in particular $U_{(\ell,k)}$ that swaps subsystems $\ell$ and $k$. Then
\[\bar\rho_k=\trace_{\bigotimes_{j\neq k}\Hi_j}(\rho)=\trace_{\bigotimes_{j\neq k}\Hi_j}(U_{(\ell,k)}\rho U_{(\ell,k)}^\dag)=\bar\rho_\ell,\]
and the reasoning can be repeated for any pair.
{\em RSC $\implies$ OSC} is immediate by definition.
States $\rho^B$ and $\rho^A$ from Example 1 provide counterexamples for the converse of the first and of the second implication, respectively.
\qed

The converse result for the relation between RSC and $\sigma$EC goes as follows.
\begin{prop}\label{prop-rsc} A state is RSC if and only if it is $\sigma$EC for all $\sigma\in{\mathfrak B}(\Hi).$ 
\end{prop}
\proof RSC implies $\sigma$EC for all $\sigma$ by definition. Conversely, assume $\rho$ to be $\sigma$EC for any $\sigma$. The reduced state for any subsystem can be written as $\bar\rho_k=\sum_\ell\trace(\rho\sigma_\ell^{(k)})\sigma_\ell$, where $\{\sigma_\ell\}$ is an orthogonal, hermitian basis of operators for ${\mathfrak B}(\Hi).$ Since by hypothesis $\trace(\rho\sigma_\ell^{(k)})$ does not depend on $k,$ then all the reduced states are identical.
\qed

A converse result for the relation between RSC and SSC cannot be set up in general. Instead, we have the following property.
\begin{prop}\label{prop-meas}
	RSC states with $\bar\rho_k$ a pure state for each $k$ are also SSC states.
	 \end{prop}
\proof If $\bar\rho_k$ is a pure state for each $k$, then necessarily $\rho = \q{\psi}\qd{\psi}$ with $\q{\psi} = \q{\psi_1} \otimes \q{\psi_2} \otimes ... \otimes \q{\psi_m}$ for some $\q{\psi_k}$, $k=1,2,...,m$. If in addition we require RSC, then we need $\q{\psi_k}\qd{\psi_k} = \bar\rho_k = \bar\rho_j = \q{\psi_j}\qd{\psi_j}$ for all $j,k$, thus $\q{\psi_k} = \q{\psi_j}$ up to an irrelevant phase factor for all $j,k$ and particle permutation indeed leaves $\q{\psi}$ invariant. \qed

The structure of the $\sigma$SMC states is further characterized in the following proposition.
\begin{prop}\label{propSYM} A state is in $\sigma$SMC if and only if, defining $\Pi_{\rm sym}=\sum_j \Pi_j^{\;\otimes m},$ it holds
\beq\label{symprop} \trace(\Pi_{\rm sym}\rho)=1,\eeq or equivalently \beq \Pi_{\rm sym}\rho\Pi_{\rm sym}  =\Pi_{\rm sym}\rho= \rho. \eeq
\end{prop}
\proof
Note that the properties $\trace(\Pi_{\rm sym}\rho)=1$ and $\Pi_{\rm sym}\rho\Pi_{\rm sym} = \Pi_{\rm sym}\rho = \rho$ are equivalent because $\Pi_{\rm sym}$ is an orthonormal projector and $\rho$ is self-adjoint positive semi-definite with unit trace.
 Assume \eqref{symprop} to hold. Along with the identities
$\Pi_j^{(k)} \Pi_{\rm sym} = \Pi_j^{(k)} \Pi_j^{\otimes m} = \Pi_j^{\otimes m}$, this gives:
\begin{eqnarray*}
\trace(\Pi_j^{(\ell)}\Pi_j^{(k)}\rho) & = & \trace(\Pi_j^{(\ell)}\Pi_j^{(k)} \Pi_{\rm sym}\rho )\\
& = & \trace(\Pi_j^{(\ell)}\Pi_j^{\otimes m}\rho )\\
 & = & \trace(\Pi_j^{(\ell)} \Pi_{\rm sym} \rho )\\
& = & \trace(\Pi_j^{(\ell)} \rho ) \,,
\end{eqnarray*}
for all $j,k,\ell$. Hence, the SMC definition \eqref{SMCeq} indeed holds.
\newline On the other hand, suppose that \eqref{symprop} does not hold. This means that $\trace((I-\Pi_{\rm sym})\rho) > 0.$ We want to show that this implies 
$$ \trace(\Pi_j^{(k)}\Pi_j^{(\ell)}\rho)  \neq  \trace(\Pi_j^{(\ell)}\rho)$$
for some $j,k,\ell$. Let us write
\[I-\Pi_{\rm sym}=\sum_{j_1,\ldots, j_m \text{ except } \{j_1=\ldots = j_m\}} \; \Pi_{j_1}\otimes \ldots \otimes \Pi_{j_m}.\]
Since $\trace((I-\Pi_{\rm sym})\rho) > 0$ implies that $\trace( \Pi_{j_1}\otimes \ldots \otimes \Pi_{j_m} \rho) > 0$ for at least one of the terms in the above sum, let us take one such term, denote the corresponding indices as $\{ \bar j_s \}$ and denote by $k, \ell$ two subsystems such that $\bar j_k\neq \bar j_\ell $ in that term. Now writing $\Pi_{\bar j_1}\otimes \ldots \otimes \Pi_{\bar j_m} =  \Pi_{\bar j_1}^{(1)} \, \Pi_{\bar j_2}^{(2)} ... \, \Pi_{\bar j_m}^{(m)}$, where all factors commute, we have:
\[\trace(\Pi_{\bar j_k}^{(k)}\Pi_{\bar j_\ell}^{(\ell)}\rho)\geq \trace(\Pi_{\bar j_1}^{(1)} \, \Pi_{\bar j_2}^{(2)} ... \, \Pi_{\bar j_m}^{(m)}\rho) >0.\]
Exploiting basic projection properties we thus get:
\begin{eqnarray*}
\trace(\Pi_{\bar j_\ell}^{(k)}\Pi_{\bar j_\ell}^{(\ell)}\rho) & \leq & \trace( (1-\Pi_{\bar j_k}^{(k)}) \Pi_{\bar j_\ell}^{(\ell)} \rho) \\
& = & \trace(\Pi_{\bar j_\ell}^{(\ell)} \rho) - \trace(\Pi_{\bar j_k}^{(k)} \Pi_{\bar j_\ell}^{(\ell)} \rho) \\
& \leq & \trace(\Pi_{\bar j_\ell}^{(\ell)} \rho ) - \trace(\Pi_{\bar j_1}^{(1)} \, \Pi_{\bar j_2}^{(2)} ... \, \Pi_{\bar j_m}^{(m)} \rho) \\
& < & \trace(\Pi_{\bar j_\ell}^{(\ell)} \rho) \, .
\end{eqnarray*}
\qed

To conclude, the following proposition connects $\sigma$SMC to the other properties.
\begin{prop}\label{prop:allto4} (a) $\sigma$SMC implies $\sigma$EC, while the converse is not true.
	\newline (b) $\sigma$SMC for a $\sigma$ with nondegenerate spectrum implies RSC, while the converse is not true.
	\newline (c) $\sigma$SMC for a $\sigma$ with nondegenerate spectrum implies SSC, while the converse is not true.
\end{prop}
\proof For (a), we have since \eqref{SMCeq} holds for all $k,\ell$:
\beqan\trace(\Pi_j^{(k)}\rho)=\trace(\Pi_j^{(k)}\Pi_j^{(\ell)}\rho)=\trace(\Pi_j^{(\ell)}\rho).\eeqan
By linearity, we thus have:
\beqan\trace(\sigma^{(k)}\rho)&=&\sum_js_j\trace(\Pi_j^{(k)}\rho)=\sum_js_j\trace(\Pi_j^{(\ell)}\rho)\\&=&\trace(\sigma^{(\ell)}\rho).\eeqan
A counterexample for the converse is state $\rho^A$ in Example 1.

Counterexamples for the converse of (b) and (c) are respectively states $\rho^B$ and $\rho^C$ in Example 1. For the direct statements, given Proposition \ref{hierarchy}, we know that if (c) is true, then (b) must be true as well. Let us then focus on (c). Take the representation of $\rho$ in the basis associated to $\sigma=\sum_j s_j \q{j}\qd{j}$, where thus $\q{j}\qd{j}=\Pi_j$, that reads
\begin{align}\label{eq:fullwrittenout}&\rho =\hspace{-3mm}\sum_{{\tiny\hspace{-0mm}\begin{array}{l}{j_1,j_2,...j_m,}\\{\,k_1,k_2,...k_m}\end{array}}}   r_{{\tiny\hspace{-2mm}\begin{array}{l}{j_1,j_2,...j_m,}\\{\,k_1,k_2,...k_m}\end{array}}} \hspace{-2mm}\q{ {j_1}, {j_2},..., {j_m}}\qd{ {k_1}, {k_2},..., {k_m}}\, . \nonumber
\end{align}
From Proposition \ref{propSYM}, the condition for $\sigma$SMC writes
$$\sum_{k,j} \, (\q{k}\qd{k})^{\otimes m}\, \rho\, (\q{j}\qd{j})^{\otimes m}  = \rho \, ,$$
so \eqref{eq:fullwrittenout} must reduce to
\begin{align}\rho = \sum_{k,j} \, p_{kj} \q{kk\ldots k}\qd{jj \ldots j} \quad\end{align}
for some $p_{kj} \in \mathbb{C}$.
It is straightforward to see that a $\rho$ of this form satisfies SSC, since any element in the sum is invariant w.r.t.~subsystem permutations.\qed


\subsection{On detecting consensus}

In the quantum setting, the relation between state and accessible information is stochastic at a fundamental level. In fact, even a ``deterministic'', maximal information state yields probabilistic outcomes for certain observables. As we are compelled to use probabilistic notions, consensus can only be inferred from  stochastic measurement records, and different types of consensus  require different types of measurement statistics.

The $\sigma$EC, requiring only equal expectations for a particular observable $\sigma$ on the different subsystems, is trivial to relate to measurement results; in particular, checking $\sigma$EC does not require to monitor correlations between measurement results on different subsystems. Checking RSC requires statistics for a basis of observables for each subsystem; as for $\sigma$EC, potential correlations between measurement results on different subsystems play no role.
On the other hand, distinguishing SSC from RSC does require to inspect correlations between measurement outcomes at different subsystems -- except for the case mentioned in Proposition \ref{prop-meas}, where SSC can be inferred directly from the special form of the reduced state of each subsystem. A proof of this is provided in Appendix \ref{app2}. For instance, considering the state $\rho^B$ of Example 1, measurements of $\sigma_z$ on the three subsystems would quickly show that the results on subsystems 2 and 3 are always perfectly correlated, and show no correlation at all with the results on the first subsystem. This difference in correlations rules out $\rho^B$ as a candidate for SSC.
The definition of $\sigma$SMC is all about correlations between measurement outcomes at different subsystems: the latter must be fully correlated for a particular observable $\sigma$. 

Positively detecting states in ${\cal C}_{\rm SSC} \setminus {\cal C}_{\sigma \text{SMC}},$ however, appears to be less obvious (except through full state tomography). Indeed, a priori it would require to check how permutations affect the full state, including not only classical correlations but also \emph{entanglement} {(see Appendix \ref{miltipartite} for a quick review)}, an intrinsically quantum type of correlation that plays a central role in quantum information theory and its applications \cite{nielsen-chuang,WildeQIT}.


\section{Quantum Evolutions for Consensus}\label{sec:model}

So far we have been concerned with discussing ``static'' properties of consensus states. However, the core of the problem is the design of dynamical interactions (or algorithms) that drive the system to consensus. In the following sections, we will establish which dynamics are needed to drive an arbitrary initial state towards a consensus state. In addition to this, as in the classical consensus problems, we shall require the final state to preserve or express some property dependent on the initial state. For example, when classical average consensus is reached, each agent locally contains (``has computed'') the average of the initial state values, which is a global property. In the quantum case as well, a goal when reaching consensus would often be to retrieve, in the final state of any local subsystem, some global information about the initial state. 


\subsection{Classical Dynamics and Locality}\label{subsec:dyns&locality:classical}

For classical consensus, the starting point is a first-order integrator dynamics for each individual agent, of the type:
\begin{equation}\label{eq:ClassicalBasis}
	x_k(t+1) = x_k(t) + u_k(t) \quad \text{or} \quad \tfrac{d}{dt}x_k(t) = u_k(t), 
\end{equation}
for $k=1,2,...,m.$ In this paper we focus on the discrete-time case. The inputs $u_k(t)$ for the agents can take different forms, but they are all based on \emph{local information}. A notion of locality is traditionally introduced by specifying a graph $G(V,E)$ whose vertices are the agents $1,2,...,m$ and where an edge $(j,k) \in E$ (unordered pair of vertices) indicates that agents $j$ and $k$ are neighbors. Then for each $k$, input $u_k$ is restricted to depend only on $x_k$ and on the states $x_j$ of agents $j$ for which $(j,k) \in E$. Mostly, $u_k$ is written as the sum of contributions from different edges, $u_k(t) = {\textstyle \sum_{j : (j,k) \in E}} \; f_{j,k}(x_k(t),x_j(t))\, .$ 
Directed graphs can be used to model directed information flow, i.e.~$E$ contains directed edges (ordered pairs of vertices) such that $(j,k) \in E$ indicates that $u_k$ depends on $x_j$, but not necessarily $u_j$ on $x_k$ (the latter being governed by another potential edge, $(k,j) \in E$). Weights $w_{(j,k)}$ can be associated to the edges to model their different relative strengths. 
Besides the \emph{given} asymmetry through edge weights, one usually forbids $u_k$ to explicitly use agent identifiers and apply a different treatment to information coming from or going towards different neighbors. This implies
\begin{equation}\label{eq:ClassicalukForm}
	u_k(t) = {\textstyle \sum_{j : (j,k) \in E}} \; w_{(j,k)} \, f(x_j,x_k) \, ,
\end{equation}
where now $f$ must be independent of $j,k$. Finally, one may request that the influence of agent $j$ on agent $k$ be exactly equivalent to the reverse influence, of agent $k$ on agent $j$. Then $u_k$ of the type \eqref{eq:ClassicalukForm} features symmetric weights, $w_{(j,k)} = w_{(k,j)}$ --- leading to an undirected weighted graph --- and the antisymmetry property $f(x,y) = -f(y,x)$ for any $x,y \in \R^n$. Such evolution would preserve the average of the state values.

More generally, locality can be defined on the basis of \emph{quasi-local operators}. Instead of considering a graph, we define a set of \emph{neighborhoods} $\mathcal{N}_j \subseteq \{1,\dots,m\}$ for $j=1,\dots,M$, and a quasi-local operator is one that leaves the states of all subsystems unchanged except those of one neighborhood $\mathcal{N}_j$.
Then a simple dynamics with local coupling would write
\begin{equation}\label{eq:ClassicalAlt}
	x(t+1) = \, {\textstyle \sum_{j=1}^M} \; V_j \, (x(t))
\end{equation}
where $V_j : \R^{mn} \rightarrow \R^{mn}$ for each $j$ is a quasi-local operator acting on the neighborhood $\mathcal{N}_j$. Treating all involved agents equivalently can be formulated as requiring that for each $j$, the quasi-local operator satisfies $P V_j = V_j P$ for all the permutations $P$ that only permute agents of $\mathcal{N}_j$. The classical consensus algorithm \eqref{eq:ClassicalBasis},\eqref{eq:ClassicalukForm} is obtained by taking two-agent neighborhoods only, and identifying each $\mathcal{N}_j$ with an edge of the graph.


\subsection{Quantum Dynamics and Locality}\label{subsec:dyns&locality:quantum}

According to Schr\"odinger's equation, {\em isolated} quantum systems evolve unitarily \cite{sakurai,nielsen-chuang}. However, unitary dynamics are not enough when we are interested in studying or engineering convergence features for a quantum system. A more general framework that includes (Markovian) open-system evolutions is offered by {\em quantum channels} \cite{kraus, nielsen-chuang}, that is, linear, completely positive (CP) and trace preserving (TP) maps from density operators to density operators $\Ei: \mathfrak{D}(\Hi^m)\rightarrow \mathfrak{D}(\Hi^m).$  It can be shown that such maps admit an {\itshape operator sum representation} (OSR), also known as {\itshape Kraus decomposition}:
\begin{equation}
\label{eq:Kraus}
\Ei(\rho)=\sum_{k=1}^{K}\, A_k\rho A_k^{\dagger} \quad \text{with} \quad \sum_{k=1}^{K} A_k^{\dagger}A_k= I
\end{equation}
where $K\leqslant (\dim(\Hi))^2$.
The representation is not unique, but all the representations can be obtained as unitary combination of the operators of a given one (see \cite[Theorem 8.2]{nielsen-chuang}). 
A CPTP map is said {\itshape unital} if $\Ei( I)= I$. Unital quantum channels can always be represented as given by random unitaries \cite{wolf-unital}. A channel belongs to this class when it admits an OSR with operators $A_k=\sqrt{p_k}U_k$, with $U_k\in {\mathfrak U}(\Hi^m)$ and $p_k\geq 0$ such that $\sum_kp_k=1$:
\[{\cal E}(\rho)=\sum_k p_k\, U_k\rho U_k^\dag. \]Such a map can be thought of as a probabilistic mixture of unitary evolutions.

Given a CPTP map $\Ei$, we can define its dual map with respect to the Hilbert-Schmidt inner product $\Ei^{\dagger}: \mathfrak{B}(\Hi)\rightarrow \mathfrak{B}(\Hi)$ through the relation:
\begin{equation}\label{eq:State/Meas-Duality}
\trace[A\,\Ei(\rho)]=\trace[\Ei^{\dagger}(A)\,\rho] \, .
\end{equation}
This dual map is still linear and completely positive, while the fact that ${\cal E}$ is trace preserving implies that ${\cal E}^\dag $ is always unital. Considering the dynamics in the dual picture, i.e.~with time-invariant states and maps acting on the observables, is called Heisenberg's picture in the physics literature and provides an equivalent description of quantum system evolution.

We now introduce {\em locality notions} for the quantum network. Consider the multipartite system introduced in Section \ref{subsec:def:quantum}: following \cite{ticozzi-QL}, we say that an operator in $\mathfrak{B}(\Hi )$ is quasi-local if it acts non-trivially only on one neighborhood $\mathcal{N}_j \subseteq \{1,\dots,m\}$:
\begin{defin}[Quantum quasi-local operator]
An operator $V$ is quasi-local with respect to a set of neighborhoods $\{ \mathcal{N}_j \; , \, j=1,2,...,M \}$, if and only if there exists $j \in \{1,2,...,M\}$ such that:
\begin{equation}
\label{eq:quasilocaloperator}
V=V_{\mathcal{N}_j}\otimes I_{\compl{{\mathcal{N}}_j}}
\end{equation}
where, with a slight abuse of notation, $V_{\mathcal{N}_j}$ accounts for the non trivial action on $\Hi_{\mathcal{N}_j}$ and $I_{\compl{\mathcal{N}_j}}=\bigotimes_{k\notin \mathcal{N}_j}I_{k}$.\\
\end{defin}


\subsection{Timing of operations and evolution types}\label{subsec:timing}

In classical consensus, an important aspect is that the graph (and the related interaction law) can be time-varying. For instance one can assume that all edges are activated for the whole time (\emph{synchronous update}), at the other extreme that they are activated one at a time, or some at each time (\emph{asynchronous update}), according to some predefined time-varying sequence or by random selection of edges. Interestingly, convergence properties for all these cases can be linked to the connectedness of the ``average graph'' \cite{Moreau2005}.

In the quantum case also this distinction can be made. The elementary dynamical interaction that we consider, replacing ``one edge'' of the classical case, is a CPTP map involving one neighborhood only:
\beq \label{eq:bblock} {\cal E}_{{\cal N}_j}(\rho)=\sum_k p_k V_k(t)\rho V_k^\dag(t), \eeq
where all the $V_k(t)\in{\mathfrak U}(\Hi^m)$ are quasi-local with respect to the neighborhood ${\cal N}_j$. One of the reasons for focusing on this class of evolutions stems directly from applications: methods for implementing unitary evolutions, as well as related unital channels with the aid of some ancillary systems, are available in a number of diverse experimental settings. On the other hand, constructing arbitrary quantum channels is a more challenging task \cite{albertini-feedback}, and can be generally done with good approximation only in the limit of fast control and/or short time scales \cite{viola-engineering}. The building block \eqref{eq:bblock} can lead to different evolutions for the whole system, depending on neighborhood selection:
\begin{itemize}
\item \emph{Random single interaction:} at each time $t$ one neighborhood ${\cal N}_{j(t)}$ is selected at random, $j(t)$ being a single-valued random variable onto the neighborhood index set.
\item \emph{Cyclic single interaction:} at each time $t$ one neighborhood ${\cal N}_{j(t)}$ is selected deterministically, for example periodically cycling between the available $j$. 
\item \emph{Random or cyclic asynchronous interactions:} similar to the previous options, but a subset of several neighborhoods is selected at each time $t$. We can request the selected neighborhoods to be disjoint or not. \emph{This choice has an effect for implementation, but not for convergence properties of our algorithm,} so we will not consider it further.
\item \emph{Synchronous interaction:} all the available interactions are activated at each time, weighted by some $q_j\geq 0$ with $\,\sum_j q_j = 1$ in order to maintain a valid trace-preserving map: 
\beq\label{timeindip}{\cal E}(\rho)=\sum_j \, q_j \, \Ei_{{\cal N}_{j}}(\rho) \, .\eeq
\item \emph{Expected evolution:} we study the evolution {\em in expectation} of the random interaction protocol which selects neighborhood $\mathcal{N}_j$ with probability $q_j$ at each $t$. Remarkably, the evolution to $\rho_{t+1}$ given $\rho_t$ then follows the same law \eqref{timeindip} as the synchronous case. Note that convergence of the expected evolution to consensus does not guarantee (at all) that a(ny) single evolution, determined by a realization of the random process $\{j(t)\}_{t\geq 0},$ would converge to consensus. Nevertheless, 
the statistics of any measurements performed at any time on the system will be exactly the same for \eqref{timeindip} as for the associated random evolution. In this sense, convergence in expectation is indistinguishable from trajectory-wise convergence. 
\end{itemize}

The last two cases involve a time-independent map. Another time-independent map is obtained if we consider cyclic interactions of period $T$ and we focus on the state at the end of every cycle:
\beq\label{cyclic}\rho_{t+T}={\cal E}_C(\rho_t)={\cal E}_{{\cal N}_{T}}\circ\ldots\circ{\cal E}_{{\cal N}_{1}}(\rho_t).\eeq

The consensus goal can now be specified formally.

Let $d(\rho_a,{\cal C})=\inf_{\rho \in{\cal C}}\|\rho_a-\rho\|,$ where ${\cal C} \subset {\mathfrak D}(\Hi)$ and $\| \cdot \|$ is any $p$-norm on ${\mathfrak B}(\Hi).$ Given a sequence of channels $\{\Ei_t(\cdot)\}_{t=0}^\infty,$ define $\hat{\Ei}_t(\rho_0)=\rho_t={\cal E}_t\circ\Ei_{t-1}\circ\cdots\circ{\cal E}_1(\rho_0).$
\begin{defin}[Asymptotic Consensus]\label{def:asy}
A sequence of channels $\{\Ei_t(\cdot)\}_{t=0}^\infty,$ is said to \emph{asymptotically achieve $\sigma$EC} if 
\begin{equation}
\label{evolution}
\lim_{t \rightarrow \infty}d(\hat{\Ei}_t(\rho_0),\mathcal{C}_{\sigma\text{EC}})=0,
\end{equation}
for all initial states $\rho_0$.
\end{defin}
The same definition holds for RSC, SSC, and $\sigma$SMC by substituting the corresponding state sets in $(\ref{evolution})$.
\begin{defin}[Asymptotic Average Consensus]\label{def:averageQcons}
We say that the sequence of channels $\{\Ei_t(\cdot)\}_{t=0}^\infty$ asymptotically achieves \emph{$S$-average $\sigma$EC} for some $S\in{\mathfrak H}(\Hi^{m})$ if it asymptotically achieves $\sigma$EC and for all $\rho_0$, it holds:
\beqa\label{eq:avgconsensus}\lim_{t \rightarrow \infty}\trace(\sigma\bar\rho_\ell(t)) &=& \lim_{t \rightarrow \infty}\trace(\sigma^{(\ell)}\rho(t))=\lim_{t \rightarrow \infty}\trace(S \rho(t))\nonumber\\ & =& ~\trace(S \rho_0)\eeqa for all $\ell\in\{1,\ldots ,m\}$. The same definition holds for $\sigma$SMC.
\newline
We say that the sequence of channels $\{\Ei_t(\cdot)\}_{t=0}^\infty$ asymptotically achieves \emph{$S$-average RSC (resp.~SSC)} if it asymptotically achieves RSC (resp.~SSC) and for $S\in{\mathfrak H}(\Hi^{m})$ {\em there exists} a $\sigma \in {\mathfrak H}(\Hi)$ such that \eqref{eq:avgconsensus} holds for all $\rho_0$.
\end{defin}

By expressing the action of quantum channels in the dual (Heisenberg) picture, it is possible to obtain a clear characterization of the dynamics that satisfy \eqref{eq:avgconsensus}.
\begin{prop}\label{dualcons}
Consider a sequence of CPTP channels $\{\Ei_t(\cdot)\}_{t=0}^\infty,$ and call $\hat \Ei_t=\Ei_t\circ\Ei_{t-1}\circ\ldots\circ\Ei_1$. 
The associated dynamics satisfies \eqref{eq:avgconsensus} if and only if 
\beq \label{dualavg} S=\lim_{t\rightarrow \infty}\hat\Ei_t^\dag(S) \quad \text{and} \quad \lim_{t \rightarrow \infty}\hat\Ei_t^\dag(\sigma^{(\ell)}) = S 
\eeq for $\ell=1,2,\ldots, m,$ 
where $\hat{\Ei}^{\dagger}_t =\Ei^{\dagger}_{1} \circ \Ei^{\dagger}_{2} \circ \dots \circ \Ei_t^{\dagger}$. \end{prop} 
\begin{proof} The conditions \eqref{dualavg} clearly imply \eqref{eq:avgconsensus}. On the other hand, if \eqref{eq:avgconsensus} holds for all $\rho_0,$ it is easy to obtain \eqref{dualavg} by duality, taking the limit inside the trace functional. \end{proof}
The first of the equalities in \eqref{dualavg} holds in particular for the natural situation where $\Ei_t^\dag(S) = S$ for all $t$.


\section{A Gossip Algorithm for Quantum Consensus}\label{sec:Algorithms}

We now propose actual interactions that drive the quantum network to average consensus. As a building block, we focus on the interaction between two subsystems while the others remain unchanged; all neighborhood-activation options build on this elementary case, as explained above.


\subsection{Another viewpoint on the Classical Gossip Algorithm}\label{subsec:Gossip}

The standard linear consensus algorithm corresponds to \eqref{eq:ClassicalBasis},\eqref{eq:ClassicalukForm} with $f(x,y) = \alpha(x-y)$. Its form with a single interaction activated at any time --- also called \emph{gossip algorithm} --- is usually described as follows~\cite{BoydGossip}. At each iteration, a single edge $(j,k)$ is selected from the set $E(t)$ of available edges at that time. The associated agents move towards each other / their mean value, according to:
\begin{eqnarray}
\nonumber	x_j(t+1) & = & x_j(t) + \alpha (x_k(t)-x_j(t))\\ & =&\nonumber  (1-\beta) x_j(t) + \beta \,\tfrac{x_j(t)+x_k(t)}{2} \\
\nonumber	x_k(t+1) & = & x_k(t) + \alpha (x_j(t)-x_k(t)) \\ & =&\nonumber (1-\beta) x_k(t) + \beta \,\tfrac{x_j(t)+x_k(t)}{2} \\
\label{eq:classical-gossip}	x_\ell(t+1) & = & x_\ell(t) \quad \text{ for all } \ell\notin \{ j,k\} \, ,
\end{eqnarray}
where $\alpha \in (0,1)$ to have meaningful results\footnote{That is, the new states are an interpolation and not extrapolation between $x_j(t)$ and $x_k(t)$.}, and $\beta = 2\alpha$. An alternative viewpoint on this behavior is that the interacting agents take a weighted average between two discrete operations: [keep your state] and [swap your state]; namely
\begin{eqnarray}
\nonumber	\left( x_j(t+1),\;x_k(t+1) \right) & = & (1-\alpha)\, \left( x_j(t),\;x_k(t) \right)\\&& + \;\alpha \, \left( x_k(t),\;x_j(t) \right) \label{eq:classical-gossipswap}\\
	x_\ell(t+1) & = & x_\ell(t) \quad \text{ for all } \ell \notin \{j,k\} \; \nonumber.
\end{eqnarray}
This latter viewpoint turns out to have a natural quantum counterpart. Working with neighborhoods, one could also apply multi-agent permutations, e.g.~
\[\begin{split}\left( x_j,\;x_k,\;x_l \right)_{t+1} = &(1-\alpha-\beta)\, \left( x_j,\;x_k,\;x_l \right)_t + \alpha \, \left( x_k,\;x_l,\;x_j \right)_t \\&+ \beta \, \left( x_l,\;x_j,\;x_k \right)_t 
\; .\end{split}\]
The following result (see e.g.~\cite{BoydGossip}) characterizes convergence to consensus with the gossip algorithm. While it is a known result, we nonetheless provide a proof that will be useful to our aim, i.e. proving the convergence of the quantum gossip algorithm, and makes our presentation more self-contained.
\begin{prop}  \label{prop:classicalgossip}
Consider $G(V,E)$ an undirected graph that is connected, i.e.~for any pair of vertices $a,b \in V$, there exists a sequence of vertices $v_0=a,v_1,v_2,...,v_{n-1},v_n=b$ such that $(v_{k-1},v_k) \in E$ for all $k=0,1,...,n$. If one step of the classical gossip algorithm \eqref{eq:classical-gossip} is applied at each time, selecting the updated edge by cyclically running through all the edges of $G(V,E)$, then the system exponentially converges to average consensus. Moreover, if the updated edge $(j,k)$ is selected randomly according to a fixed probability distribution $\{ q_{j,k} \}$, with all $q_{j,k} > 0$, then asymptotic average consensus is ensured with probability one, in the sense that: for any $\delta,\varepsilon > 0$, there exists a time $T>0$ such that 
$$\mathbb{P}\left[\|x_k(T)-\bar{x})\|^2 > \varepsilon \|x_k(0)-\bar{x}\|^2\right] \; < \delta \, ,$$
where $\mathbb{P}$ denotes the probability measure induced by the randomization, $\|x\|^2=\sum_k\, x_k^T x_k$ and $\bar{x} = \tfrac{1}{m} \sum_k x_k(t) = \sum_k x_k(0)$ for any choice of the edges.
\end{prop}
\proof
We denote $x^T x = \Vert x \Vert^2$ for short and $\#E$ the number of edges in $G(V,E)$. At any step of the gossip algorithm, $W:= \tfrac{1}{2m} \sum_{k,j} \Vert x_k-x_j \Vert^2 = \sum_k  \Vert x_k-\bar{x} \Vert^2$ can only remain unchanged (if the two nodes of the selected edge have the same value) or decrease (as soon as an edge with different node values is selected). A Lyapunov argument on $W$ then shows that the system must asymptotically converge to average consensus when the edges of a connected graph are selected in a cyclic way. Since the map associated to one full cycle of edge selections is linear and time-invariant, this convergence is exponential. 
For such convergence to be possible, there must exist some $\lambda>0$ and integer $M>0$ such that $W(T) \leq W(0) \lambda$ if the edge choice between $t=0$ and $t=T = M \#E$ corresponds to $M$ cycles of gossip iteration. When edges are selected randomly, any particular sequence of $b$ consecutive edge selections has a probability greater than $ \bar{q}^b > 0$ to appear at least once during any time interval of length at least $b$, where $\bar{q} = \min_{(j,k) \in E} q_{j,k}$. In particular, if we target $W(T) < \varepsilon W(0) = \lambda^r W(0)$, we can say that there is a probability at least $\bar{q}^{rM \#E}$ to select $r$ times a succession of $M$ cyclic interactions between $t=t_0$ and $t=t_0+rM\#E$. If this happens once, any preceding or following edge choice can only improve $W$ (because of our first statement in this proof). We conclude by noting that over a time interval $brM \#E$, there is then a probability $< (1-\bar{q}^{rM\#E})^b$ to have never selected $r$ times a succession of $M$ cyclic interactions, and thus potentially miss $W(T) < \varepsilon W(0)$; the probability that this happens can be made arbitrarily small by taking $b$ (thus $T$) sufficiently large.
\qed


\subsection{Quantum Gossip Interactions}

Let us introduce a way to implement gossip-type interactions in a fully quantum way. In a controlled quantum network, one can typically engineer unitary transformations that implement the ``identity" evolution and the swapping of two neighboring subsystem states; let us denote the latter operator by $U_{(j,k)}$ for swapping subsystems $j$ and $k$.
By conditionally associating these two operations to orthogonal states $\q{\xi_I}$ and $\q{\xi_S}$ of an ancillary quantum two-level system and starting that auxiliary system in the state $\rho_{\xi} = (1-\alpha)\, \q{\xi_I}\qd{\xi_I} + \alpha \q{\xi_S}\qd{\xi_S}$, the joint state would evolve as:
$$ \rho \otimes \rho_\xi \, \rightarrow \, (1-\alpha)\, \rho \otimes \q{\xi_I}\qd{\xi_I} + \alpha\, U_{(j,k)} \rho U_{(j,k)}^\dagger \otimes \q{\xi_S}\qd{\xi_S} \, .$$ 
{\em The evolution of the composite system is thus described by a single unitary transformation.\footnote{The unitary transformation being given by $I\otimes \q{\xi_I}\qd{\xi_I}+ U_{(j,k)} \otimes \q{\xi_S}\qd{\xi_S}$.}} Taking the partial trace over the ancillary system, we obtain as evolution for the quantum network a quantum gossip interaction implementing the quantum channel:
\begin{equation}\label{alg:gossip}
	\rho(t+1) = {\cal E}_{j,k}(\rho(t))=(1-\alpha)\, \rho(t) + \alpha \, U_{(j,k)}\rho(t)U_{(j,k)}^\dag  \; ,
\end{equation}
with $\alpha \in (0,1).$ The optimal quasi-local mixing is obtained with $\alpha = 1/2$. 

To develop our analysis, it will be convenient to introduce the {\em graph ${\cal G}$ associated to the multipartite system}: its nodes $1,\ldots,m$ correspond to the ``physical'' subsystems, the edge $(j,k)$ is included if the subsystems $j$ and $k$ have a non-zero probability to interact.

  
\subsection{Convergence to Consensus}

We study convergence under three types of gossip dynamics: cyclic interaction, expectation of random interaction, and trajectory-wise for the random interaction. In all these cases, quantum gossip can be described by unital CPTP maps. We begin by recalling a characterization of the fixed points of such maps (see e.g.~\cite{viola-IPSlong}).

\begin{prop}\label{fixedpoint}
Let $\{V_i\}_{i=1}^K$ the Kraus decomposition of a {\em unital} CP map $\Ei(\cdot)$ and define:
\begin{equation}
\label{ }
\mathcal{A}_{\Ei}=\{X \in \mathfrak{B}(\Hi^m)\,\vert\,[X, V_i]=0~\,\forall\, i=1,\dots,K\}\, .
\end{equation}
Then $\bar{X} \in \mathfrak{B}(\Hi^m)$ is a fixed point of $\Ei$, i.e. $\Ei(\bar{X})=\bar{X}$, if and only if $\bar{X} \in \mathcal{A}_{\Ei}$. \hfill $\square$
\end{prop}

This helps determine the set of fixed points for the CP maps of interest in quantum gossip.

\begin{lemma}\label{perm}
Let $U_{(j,k)}$ denote the pairwise swap operation of subsystems $(j,k)$ on $\Hi^m$. If the graph ${\cal G}$ associated to the system is connected, then the set of fixed points of any CP unital map of the form
\begin{equation}\begin{split}
\label{mapL}
&\Ei(X) = q_0 \, X  + \sum_{(j,k) \in E}\, q_{j,k} \, U_{(j,k)} X U_{(j,k)}^\dagger \;\;,\\& \text{ with } q_0 + \sum q_{j,k}=1\, , \;\;\; q_0,\,\{q_{j,k}\}>0
\end{split}\end{equation}
coincides with the set of permutation-invariant operators.
\end{lemma}
\proof According to Proposition \ref{fixedpoint} above, the fixed points are the $X$ satisfying $X U_{(j,k)} = U_{(j,k)} X$, or equivalently $U_{(j,k)}^\dagger X U_{(j,k)} = X$. The latter expresses that $X$ is invariant with respect to pairwise swaps on all the graph edges. It is well known that sequences of pairwise swaps on the edges of a connected graph generate the full set of permutations on the set of nodes, and so we get the conclusion.
\qed

The following lemma shows how the contribution of the identity, i.e. the trivial permutation, in the CP map plays a crucial role in the proof of convergence.
\begin{lemma}\label{lem:nocycles}
	 If a CP map $\Ei$ admits an OSR with a term $V_1=\sqrt{\alpha}\, I$ with $\alpha>0$, then viewing it as a linear map on ${\mathfrak B}(\Hi^m)$ its only modulus-one eigenvalue can be one.  
\end{lemma}
\proof If $\Ei$ is a CPTP map it is a contraction in trace norm \cite{nielsen-chuang,alicki-lendi}, so its eigenvalues $\lambda_k$ belong to the closed unit disk. By virtue of the Kraus-Stinespring representation theorem (see e.g. \cite{kraus}), also ${\cal F} = \frac{1}{1-\alpha} (\Ei - \alpha I)$ is CPTP and thus has eigenvalues $\mu_k$ in the closed unit disk. Therefore the eigenvalues $\lambda_k = (1-\alpha) \mu_k + \alpha$ of $\Ei = (1-\alpha) {\cal F} + \alpha I$ in fact belong to the circle of radius $(1-\alpha)$ centered at $\alpha$, which is strictly inside the unit circle except for a tangency point at $1 \in \C$, see Fig.~\ref{fig:egv}. \qed
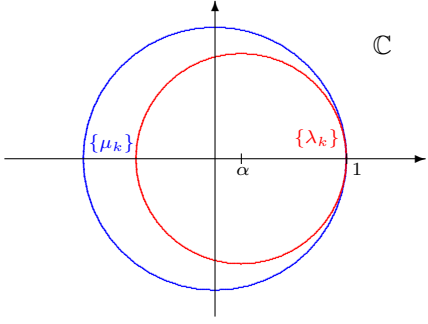
\begin{figure}
\begin{center}
	\setlength{\unitlength}{0.7mm}
	\begin{picture}(80,60)(0,5)
		\put(40,35){\color{blue}\bigcircle{50}}
		\put(45,35){\color{red}\bigcircle{40}}
		\put(55,38){\scriptsize\color{red}$\{ \lambda_k \}$}
		\put(16,37){\scriptsize\color{blue}$\{ \mu_k \}$}
		\put(0,35){\vector(1,0){80}}
		\put(40,5){\vector(0,1){60}}
		\put(70,55){$\C$}
		\put(65,34){\line(0,1){2}}
		\put(66,32){\scriptsize$1$}
		\put(45,34){\line(0,1){2}}
		\put(44,32){\scriptsize$\alpha$}
	\end{picture}
\end{center}
	\caption{Delimitation of the (closed) domains for the eigenvalues $\{ \mu_k \}$ of $\cal F$ (blue) and $\{ \lambda_k \}$ of $\Ei$ (red) [color online].}\label{fig:egv}
\end{figure}
In other words, Lemma \ref{lem:nocycles} excludes eigenvalues of unit norm different from $+1$, which are the ones that would cause limit cycles.

By combining the above properties, we get the following convergence result for quantum gossip. It shows that $S$-average SSC can be attained for global operators that are the permutation-invariant average of local ones; this is similar to classical gossip, where distributed computation of the average of individual states actually gives access to the value of any linear permutation-invariant function of these states.
\begin{thm} \label{thm:convergence}
If the graph associated to possible interactions is connected, then the quantum gossip algorithm \eqref{alg:gossip} ensures global convergence towards SSC:
\newline - deterministically, when the edges on which a gossip interaction occurs at a given time are selected by periodically cycling, in any predefined way, through the set of edges;
\newline - in expectation, when the edges on which a gossip interaction occurs at a given time are selected randomly from a fixed probability distribution $\{ q_{j,k} > 0 \vert \sum_{(j,k)\in E} q_{j,k} = 1 \}$;
\newline - with probability one on any trajectory, with the same edge-selection strategy of the previous point. Explicitly, there exists a state $\rho_* \in \mathcal{C}_{SSC}$ for which for any $\delta,\varepsilon > 0$, there exists a time $T>0$ such that
$$\mathbb{P}[\, \trace((\rho(T)-\rho_*)^2) \,  > \varepsilon \,] \; < \delta \, .$$
In any of the cases above, the state towards which the system converges is:
\begin{equation}
\label{asymptoLS}
\rho_* = \frac{1}{m!}\, \sum_{\pi \in \mathfrak{P}} \, U_{\pi}\rho_0U_{\pi}^{\dagger}\, .
\end{equation}
Furthermore, $S$-average SSC is attained if and only if $S\in\mathfrak{H}(\Hi^{\otimes m})$ can be written, for some $\sigma\in{\mathfrak H}(\Hi),$ in the form:
\begin{equation}
\label{sumsim}
S=\frac{1}{m}\sum_{i}^m\sigma^{(i)}.
\end{equation}
\end{thm}
\proof
First notice that all the operators in the OSR of the map \eqref{mapL} are self-adjoint. This implies that permutation-invariant observables $S$ are fixed points for the associated dual map, and hence for the gossip iteration associated to any edge $(j,k)$ and $\forall\,\rho\,\,$:
\begin{equation}
\label{firstL}
U_{(j,k)} S U_{(j,k)}^\dagger = S \; \Rightarrow \;\trace[\Ei_{j,k}(\rho) S] = \trace[\rho \Ei_{j,k}^\dagger(S)] =\trace[\rho  S]  \;\; .
\end{equation}
 
For the {\em cyclic evolution} map $\Ei_C,$ we notice that all the simple two-subsystem swaps are still present with a weight different from zero in the OSR of the cyclic map \eqref{cyclic}, thanks to the presence of the identity in the OSR of each gossip interaction step. Therefore by Lemma~\ref{perm} the fixed points are the permutation-invariant operators. Now consider the dynamics associated to $\Ei_C$ as a linear, time-invariant map acting on the subspace of hermitian matrices. 
From Lemma~\ref{lem:nocycles} and the fact that the time-invariant linear map leaves ${\mathfrak D}(\Hi^m)$ invariant (excluding unstable Jordan blocks),
we have that all the {\em modes} of the LTI system are asymptotically stable except those corresponding to the fixed-point set, namely the permutation-invariant set: every initial state converges to a fixed point $\rho_\infty$ in this set. Thus the SSC set is globally asymptotically stable, and in fact exponentially stable since the map is linear. Let us now prove that $\rho_\infty$ has the form \eqref{asymptoLS}. For all {\em permutation invariant} $X$, from \eqref{firstL} we have that:
\begin{equation}
\label{ }
\trace[X\Ei_C(\rho_0)]=\trace[X \rho_0]\,\,\,\,\,\,\,\forall\,\,t.
\end{equation}
Combining the latter with the fact that $\rho_\infty$ is permutation-invariant, that the set of all permutations is self-adjoint, and using \eqref{permutinvL}, we get for arbitrary  $Q \in \mathfrak{H}(\Hi^m)$:
\begin{eqnarray}
\trace[Q\rho_\infty]&=&\trace[Q\frac{1}{m!}\sum_{\pi \in \mathfrak{P}}U_{\pi}\rho_\infty U_{\pi}^{\dagger}]
\nonumber\\
&=&\trace[\frac{1}{m!}\sum_{\pi \in \mathfrak{P}}U_{\pi}Q U_{\pi}^{\dagger}\rho_\infty]
\nonumber\\
&=&\trace[\frac{1}{m!}\sum_{\pi \in \mathfrak{P}}U_{\pi}Q U_{\pi}^{\dagger}\rho_0]\nonumber\\
&=&\trace[\frac{1}{m!}\sum_{\pi \in \mathfrak{P}}Q U_{\pi}\rho_0 U_{\pi}^{\dagger}] \nonumber .
\end{eqnarray}
This implies that indeed $\rho_\infty = \rho_*$ as described in the statement. 

For the {\em expectation of random evolution}, the CPTP map $\Ei$ is exactly of the form of Lemma \ref{perm} and the same reasoning can be repeated.

For the {\em random trajectory evolution}, we repeat a proof similar to that of Proposition \ref{prop:classicalgossip}. Since ${\cal E}$ for a single interaction is linear, \emph{self-adjoint}, with eigenvalues in the closed unit disk, it is a contraction for the Frobenius norm distance $\trace((\rho_A-\rho_B)^2)$ between any two states $\rho_A,\rho_B \in \mathfrak{D}(\Hi^m)$. Indeed, ${\cal E}$ has non-increasing orthonormal modes, so by writing any operator $X \in \mathfrak{H}(\Hi^m)$ in the modal basis we directly get $\trace(\Ei(X)^\dagger \,\Ei(X)) \leq \trace(X^\dagger X)$; taking $X=\rho_A-\rho_B$ yields the contraction. This is exactly analogous to the non-increasing Euclidean norm $x^T x = \Vert x \Vert^2$ under a classical consensus interaction with an \emph{undirected} graph, and the related contraction of $\Vert x_A-x_B \Vert^2$. Now taking in particular $\rho_A = \rho$ and $\rho_B = \rho_*$, we get that the Frobenius distance from $\rho$ to $\rho_*$ can never increase. Moreover, by transitivity of the permutation operators, $\frac{1}{m!}\, \sum_{\pi \in \mathfrak{P}} \, U_{\pi}\rho U_{\pi}^{\dagger} = \frac{1}{m!}\, \sum_{\pi \in \mathfrak{P}} \, U_{\pi}\rho_0 U_{\pi}^{\dagger}= \rho_*$ for any $\rho$ along the trajectory of the gossip algorithm.
Now given the convergence under cyclic evolution, there must exist some $\lambda<1$ and integer $M>0$ such that 
$$\trace((\Ei_C^M(\rho)-\rho_*)^2) \leq \lambda \trace((\rho-\rho_*)^2) \,$$
for any $\rho$ for which $\frac{1}{m!}\, \sum_{\pi \in \mathfrak{P}} \, U_{\pi}\rho U_{\pi}^{\dagger}  = \rho_*$. The proof then concludes along the same lines as Proposition \ref{prop:classicalgossip}, namely the probability to obtain an edge sequence which includes successions of $M$ cyclic evolutions a sufficiently large number of times to have $\varepsilon$-convergence, gets arbitrarily close to $1$ if we wait long enough. 

Finally let us prove that we attain $S$-average consensus if and only if $S$ can be decomposed as in $(\ref{sumsim})$. 
We know from the first part of the proof that all permutation-invariant observables $S$ are fixed points for the associated dual map $\Ei_t^\dag$.
Then according to Proposition $\ref{dualcons}$ we have $S$-average consensus if and only if there exists a local observable $\sigma$ such that:  
 \begin{equation}\label{ }
\quad \lim_{t \rightarrow \infty}\hat\Ei_t^\dag(\sigma^{(\ell)}) = S 
\end{equation}
for $\ell=1,2,\ldots, m.$ Because of \eqref{firstL} and \eqref{asymptoLS}, by duality we have that for every local operator $\sigma^{(\ell)}$:
\begin{equation}
\label{ }
 \lim_{t \rightarrow \infty}\hat\Ei_t^\dag(\sigma^{(\ell)}) =\frac{1}{m!}\sum_{\pi \in \mathfrak{P}}U_{\pi}^\dag\sigma^{(\ell)}U_{\pi}=\frac{1}{m}\sum_{i=1}^m\sigma^{(i)}.
\end{equation}
Therefore, we can achieve $S$-average consensus if and only if $S$ is a   permutation invariant operator that can be decomposed as in $(\ref{sumsim})$.
\qed

{\em Remark:} This shows that the mean value of a (global) observable $S=\frac{1}{m}\sum_\ell \sigma^{(\ell)}$, with arbitrary $\sigma$, can be asymptotically retrieved from the state of any single subsystem after having applied one of the quantum gossip algorithms. 

On the other hand, unlike for classical consensus, there are permutation-invariant operators that do not attain $S$-average consensus, because they cannot be written in the form \eqref{sumsim}. This is the case among others if $S$ is orthogonal to the linear span of all the local observables. For instance if~$\tilde S=\sigma_z^{\otimes^m}$, given the orthogonal basis $\{\sigma_k\}_{k=0,x,y,z}$ for ${\mathfrak B}(\Hi)$, we have:
\begin{displaymath}
\trace[\tilde S\, \sigma^{(\ell)}_k]=0\quad \forall \, k \in \{0,x,y,z\} \text{ and } \forall \, l \in \{1,\dots,m\}.
\end{displaymath}
Therefore $\tilde S$ cannot be written in the form \eqref{sumsim}, hence although $\tilde S$ is conserved by the gossip algorithm, the latter cannot lead to $\tilde S$-average consensus in the sense of Definition \ref{def:averageQcons}.

As already mentioned, the convergence speed for the random case can be quite low, and a faster map would be obtained by effectively taking a mixture of all possible updates at each time. This can be attained by suitably selecting edges through an auxiliary quantum system, acting as an independent quantum ``coin'' (in the jargon of quantum random walks) at each time. This would represent a fully quantum implementation of a \emph{synchronous consensus update}.


\subsection{Classical equivalent to observable consensus dynamics}

We next show how the quantum gossip algorithm \eqref{alg:gossip} in fact implements in a quantum fashion the classical gossip as we restrict to $\sigma$EC. According to Definition \ref{def:sigmaEC}, a quantum state $\rho$ belongs to ${\cal C}_{\sigma\text{EC}}$ if:
\begin{equation}
\label{ }
\trace[\sigma^{(1)}\rho]=\ldots =\trace[\sigma^{(m)}\rho].
\end{equation}
In view of this, it seems reasonable to attempt a convergence study of the algorithm \eqref{alg:gossip} directly in terms of the evolution of the expectation values of the $\sigma^{(\ell)}$ operators. This is not possible for arbitrary quantum evolutions, since a quantum state is far from fully specified by a {\em single} set of commuting observable expectations, and different states with the same expectation may lead to very different evolutions. However, our quantum gossip algorithm remarkably allows us to write a model for the average dynamics of the $\sigma^{(\ell)}$ in closed form. More precisely, let us define $
z_\ell(t):=\trace[\Ei^t(\rho_0)\sigma^{(\ell)}]=\trace[\rho_t \sigma^{(\ell)}] \;.$
Note that for one subsystem swap $U_{(j,k)}$, we have:
\begin{equation}
\label{eq:swapeffect}
\trace[\sigma^{(\ell)}U_{(j,k)}\, \rho \, U_{(j,k)}^{\dagger}]= \begin{cases} 
z_\ell & \text{if } \ell\notin \{j,k\} \\ 
z_k & \text{if } \ell=j \\
z_j & \text{if } \ell=k \, .
\end{cases}
\end{equation}
According to \eqref{eq:swapeffect} and \eqref{alg:gossip}, the random gossip algorithm update yields, with probability $q_{j,k}$, i.e.~when the edge $(j,k)$ is selected:
\begin{eqnarray*}
(z_j(t+1),z_k(t+1)) & \hspace{-3mm}= & \hspace{-2mm}(1-\alpha)(z_j(t),z_k(t))+\alpha(z_k(t),z_j(t)) \\
z_\ell(t+1) &\hspace{-3mm} = & \hspace{-2mm}z_\ell(t) \quad \text{ for all } \ell \notin \{j,k\} \, .
\end{eqnarray*}

This last expression has exactly the same form as the classical gossip algorithm \eqref{eq:classical-gossipswap}. Therefore, Proposition \ref{prop:classicalgossip} readily implies:
\begin{cor}
	Under all the various edge selection strategies for quantum consensus algorithm \eqref{alg:gossip}, the $z_\ell(t)$, $\ell = 1,2,...,m$ asymptotically converge towards the unique configuration:

$$ \lim_{t\rightarrow \infty}z_\ell(t) = \frac{1}{m} \sum_{k=1}^m z_k(0) \quad \text{ for all } \ell \in \{1,2,...,m\} \, .$$ \hfill $\square$	
\end{cor}

We remark that this only proves average \emph{$\sigma$-Expectation Consensus} of the quantum gossip algorithm, while our previous Theorem \ref{thm:convergence} shows that the algorithm in fact ensures the stronger average \emph{Symmetric State Consensus}.


\subsection{Gossip algorithm example}\label{ssec:GossipExample}

In this section we briefly discuss the evolution induced by random quantum gossip interactions \eqref{alg:gossip} on a four-qubit network whose associated graph is a path\footnote{I.e.~the available neighborhoods, labeling the subsystems as $\{1,2,3,4\},$ are $\{1,2\}$, $\{2,3\}$ and $\{3,4\}$.}. We observe its convergence toward average $\sigma$EC, average RSC and average SSC. In particular we consider as a ``target'' global observable: 
\begin{equation}
\label{ }
S=\frac{1}{4}\left( \sigma_z^{(1)}+ \sigma_z^{(2)}+ \sigma_z^{(3)}+ \sigma_z^{(4)}\right)\, .
\end{equation}
Let the initial state be:
\begin{equation}
\label{ }
\rho=| 1,0,1,0\rangle \langle1,0,1,0|,
\end{equation}
which is pure, and does not satisfy any of the consensus definitions provided in Section \ref{sec:Definition}. 

\noindent By Theorem $\ref{thm:convergence}$ we have that the state asymptotically converges to:
\begin{equation}
\label{c-cons}
 \begin{aligned}
\rho_\infty=&\lim_{t \longrightarrow \infty}\rho(t) 
= \frac{1}{3!} \, \sum_{\pi \in \mathfrak{P}} \, U_{\pi}\rho_0U_{\pi}^{\dagger} \\ 
= \frac{1}{6}&\;\;\;(| 1,1,0,0\rangle \langle1,1,0,0|+| 1,0,1,0\rangle \langle1,0,1,0|\\&+| 1,0,0,1\rangle \langle1,0,0,1|
 +| 0,1,1,0\rangle \langle0,1,1,0|\\&+| 0,1,0,1\rangle \langle0,1,0,1|+| 0,0,1,1\rangle \langle0,0,1,1|).
\end{aligned}
\end{equation}
This expression is clearly invariant under all the subsystem permutations, i.e. $\rho_\infty$ is in SSC, and therefore also in RSC and $\sigma$EC for all $\sigma$. The expectation value of $S$ is preserved at any step, and by Theorem \ref{thm:convergence} the algorithm drives the system to $S$-average consensus, with $\sigma =\sigma_z$.

However, $\rho_\infty$ is not in $\sigma$SMC for any $\sigma \neq \alpha I$. Indeed, according to Proposition \ref{propSYM}, $\rho_\infty$ is in $\sigma$SMC if and only if $\trace[\rho_\infty \Pi_{sym}]=1$. Now let $\{\Pi_i\}_{i=1}^{6}$ denote the orthonormal rank-one projectors in \eqref{c-cons} and define the orthonormal projector $\bar\Pi=\sum_i\Pi_i$, such that 
$\rho_\infty=\frac{1}{6}\sum_{i=1}^{6}\Pi_i=\frac{1}{6}\bar\Pi \, .$ We then get 
\begin{equation}
\label{ }
\trace[\rho_\infty \Pi_{sym}]=\frac{1}{6}\trace[\sum_{i=1}^{6}\Pi_i\Pi_{sym}]=\frac{1}{6}\trace[\Pi_{sym}\bar\Pi].
\end{equation}
This last expression is equal to $1$ if and only if $\trace[\Pi_{sym}\bar\Pi]=6.$ However, excluding the trivial case $\sigma = \alpha I$, for qubit networks $\Pi_{sym}$ is always a two dimensional projector, so $\trace[\Pi_{sym}\bar\Pi]\leq 2$. Hence $\rho_\infty$ cannot be in $\sigma$SMC for any non-trivial $\sigma$.

Figure \ref{4qubitRan1} shows the evolution of the expectation values of the local and of the global observables related to $\sigma_z$ as the iterations proceed for one run. The edges are selected at random with uniform probability, and the mixing parameter $\alpha$ is taken to be $1/2$. With this particular choice, the reduced density operators of two subsystems that have just interacted are equal; this explains why a maximum of three points are visible on the graph at any time.
The plot shows that asymptotically the expectation of the local observables $\sigma_z$ tend to the expectation value of the global observable $S$, while the expectation value of $S$ is preserved at each step.

\begin{figure}
\begin{center}
\includegraphics[width=9cm]{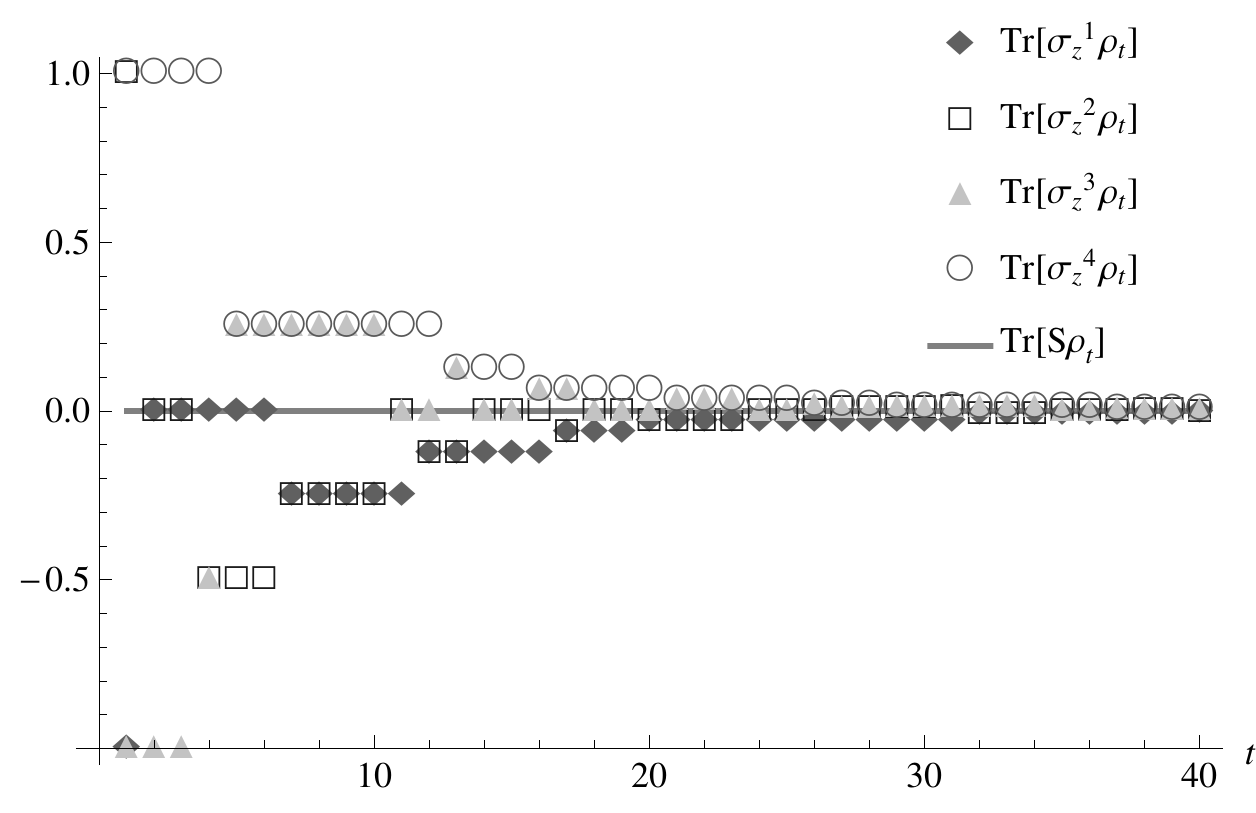}
\caption{Evolution toward $\sigma$-Expectation Consensus for a four-qubit network arranged in a path graph.}
\label{4qubitRan1}
\end{center}
\end{figure}


\section{Conclusions and Research Directions}\label{sec:conclusions}

In this paper we develop a general framework for posing and studying consensus problems in the quantum domain. In particular we provide various operationally-motivated generalizations of a ``consensus state'' to quantum systems -- namely $\sigma$-expectation consensus, reduced state consensus, symmetric state consensus, and single $\sigma$-measurement consensus -- and establish their hierarchy. We highlight at each step the symmetry considerations underlying the results, making explicit connection with the usual multi-agent consensus problem. The developments could be adapted in particular to obtain a ``consensus on probabilities'' framework for classical systems.
With respect to the existing work on non-commutative consensus \cite{SSRconsensus}, our approach follows the analogy with the classical setting as closely as possible, maintaining an operational viewpoint and working with a multipartite system (a quantum network). We propose and analyze a quantum gossip-type algorithm that asymptotically prepares symmetric-state consensus states while preserving the expectation of \emph{any permutation invariant observable}.

A number of questions remain open. Among these, we believe that it would be particularly interesting to further explore the link between single $\sigma$-measurement consensus states and entangled states, and to determine if, and under which conditions, it is possible to achieve this type of consensus with a distributed algorithm. This could potentially lead to a class of algorithms that prepare entangled states in a robust and distributed way. Another interesting point is to assess the potential of devising continuous-time quantum consensus algorithms. This could build on some sort of ``continuous swapping'' Hamiltonian dynamics and lead to connections with physically relevant many-body Hamiltonians and dynamics.

Lastly, let us remark that in this paper we proposed a quantum algorithm in which the gossip-type interactions are selected in a classical way. The potential advantage of a fully quantum implementation, along with its connection to quantum random walks and Markov chain mixing properties \cite{kempe-randomwalks,szegedy-randomwalks}, is definitely worth further investigation.
\vspace{3mm}


\appendix

\subsection{Description of quantum systems and notations}\label{app:notations}\label{miltipartite}

\subsubsection{Quantum systems basics}
This paper considers finite-dimensional quantum systems. Their mathematical description starts by considering a finite dimensional complex Hilbert space $\Hi\simeq \C^d.$ The (Dirac's) notation $\ket{\psi}$ denotes an element of $\Hi$ (called a {\em ket}), while $\bra{\psi}= \q{\psi}^\dag$ is used for its dual (a {\em bra}), and $\braket{\psi}{\varphi}$ for the associated inner product. We denote the set of linear operators on $\Hi$ by $\mathfrak{B}(\Hi)$.  The adjoint operator $X^\dag \in \mathfrak{B}(\Hi)$ of an operator $X \in \mathfrak{B}(\Hi)$ is the unique operator that satisfies $(X \q{\psi})^\dag \, \q{\chi} = \qd{\psi} \, (X^\dag \q{\chi})$ for all $\q{\psi},\q{\chi} \in \Hi$. We then denote $\mathfrak{H}(\Hi)$ the subset of $\mathfrak{B}(\Hi)$ of self-adjoint operators, and $\mathfrak{U}(\Hi) \subset \mathfrak{B}(\Hi)$ the subset of unitary operators. The natural inner product in $\mathfrak{B}(\Hi )$ is the Hilbert-Schmidt product $\langle X,Y \rangle=\trace(X^\dag Y),$ where $\trace$ is the usual trace functional (which is canonically defined in a finite dimensional setting). We denote by $I$ the identity operator.
Working in a finite dimensional setting, we often consider vectors and operators as represented by complex matrices of suitable dimensions: $\ket{\psi}\in\Hi\simeq\C^d$ are represented by column vectors, so $\bra{\phi}\in\Hi^\dag\simeq\C^d$ are row vectors; $X\in\mathfrak{B}(\Hi )\simeq\C^{d\times d}$ are $d\times d$ complex matrices, the adjoint $X^\dag$ is the transpose conjugate of $X$, self-adjoint and unitary properties carry over to the associated matrices.

In statistical quantum theory, the state of a quantum system is represented by a {\em density operator} $\rho$, that is any self-adjoint, positive semi-definite operator with trace one. We denote the convex set of these operators (the state space) by $\mathfrak{D}(\Hi )$. The extreme points of this set, namely the rank-one operators $\rho = \q{\psi}\qd{\psi}$ with $\q{\psi} \in \Hi$ and $\qd{\psi} \psi \rangle = 1$, are called \emph{pure states}.

A {\em projective} (or von Neumann)  observation, or measurement, of a quantum system is characterized by a so-called {\em observable}, that is a self-adjoint operator $\sigma \in \mathfrak{H}(\Hi )$, see e.g.~\cite{sakurai}. Its spectral decomposition $\sigma=\sum_j s_j \Pi_j$ determines the possible outcomes $\{s_j\}$ of the measurement, and the projectors $\Pi_j$ determine the associated state of the quantum system after the stochastic measurement: having state $\rho$ before the measurement, the latter's outcome will be $s_j$ with probability $\PP_{\rho}(\Pi_j) = \trace(\Pi_j\rho) =: p_j$; and if the $j$-th outcome is measured, then the state after the measurement is $\rho|_j= \Pi_j\rho\Pi_j \, / \, p_j$. The probability to observe $s'_k$ in a subsequent measurement of $\sigma' = \sum_k s'_k \Pi'_k$, where the $\Pi'_k$ do not necessarily commute with the $\Pi_j$, is then:
\[\mathbb{P}_{\rho|_j}(\Pi'_k)=\trace(\Pi'_k\Pi_j\rho\Pi_j)/ \, p_j.\]
From this it follows that the probability of observing the {\em ordered} sequence of two events first $s_j$, then $s'_k$, given the initial $\rho$, is
\[\mathbb{P}_{\rho}(\Pi_j,\Pi'_k)=\trace(\Pi'_k\Pi_j\rho\Pi_j).\]
If $\Pi_j$ and $\Pi'_k$ do not commute, a different ordering in a sequence of measurements can change the resulting probability. If $\Pi_j$ and $\Pi'_k$ do commute, and only then, the joint probability of observing $s_j,s_k$ is independent of the measurement order for all $\rho$, and simplifies to \[\mathbb{P}_{\rho}(\Pi_j,\Pi'_k)=\trace(\Pi'_k\Pi_j\rho).\]

\subsubsection{Multipartite systems and partial trace} For simplicity, we present the interaction of two quantum systems; the case of $n>2$ systems is easily obtained by iteration. If two quantum systems, with associated Hilbert spaces $\Hi_1$ and $\Hi_2$ respectively, are taken together to form a larger bipartite quantum system, the Hilbert space $\Hi_{1,2}$ associated to the composite quantum system is the tensor product of the individual quantum subsystem Hilbert spaces, $\Hi_1\otimes \Hi_2$. 

Let $\{| \psi_k\rangle\}_{k=1}^{d_1}$ and $\{|\phi_l \rangle\}_{l=1}^{d_2}$ be orthonormal bases for $\Hi_1$ and $\Hi_2$ respectively, then an orthonormal basis for $\Hi_{1,2}$ can be written as:
\begin{equation}
	\{|\psi_k\rangle \otimes |\phi_l \rangle\}_{k,l=1}^{d_1,d_2} \; ,
\end{equation}
from which we get that $\dim(\Hi_{1,2})=\dim(\Hi_1)\dim(\Hi_2)=d_1d_2$. We use the short notation $|\psi, \phi \rangle:=|\psi\rangle \otimes |\phi \rangle$ for any $\q{\psi} \in \Hi_1$ and $\q{\phi} \in \Hi_2$. The composite Hilbert space is naturally endowed with the inner-product
$
\langle u_1,u_2 |v_1,v_2 \rangle:= \langle u_1 |v_1 \rangle \langle u_2 |v_2 \rangle \, .   
$
A representation and basis for operators in $\mathfrak{B}(\Hi_{1,2})$ is derived from its vector counterpart in the standard way.
In particular, given two operators $X_1 \in \mathfrak{B}(\Hi_1)$ and $X_2 \in \mathfrak{B}(\Hi_2)$, one can define $X_1\otimes X_2\in \mathfrak{B}(\Hi_{1,2})$ as the linear operator such that $\forall \q{u_1} \in \Hi_1\, , \; \q{u_2} \in \Hi_2$:
\begin{equation}
\label{ }
X_1\otimes X_2(|u_1\rangle \otimes |u_2 \rangle)=X_1| u_1\rangle \otimes X_2|u_2 \rangle \;\;\;  .
\end{equation}
If two operators are in the form $X_1\otimes\mathcal{I}_2$ and $ \mathcal{I}_1\otimes X_2$, i.e.~they act non-trivially only on different parts of the multipartite system, then they commute for any $X_1$ and $X_2$. It is worth noting that in matrix representation, the tensor product corresponds to the Kronecker product. 

The partial trace over $\Hi_1$ is a linear map:
\begin{equation}
\label{ }
\trace_{\Hi_1}:\mathcal{B}(\Hi_1\otimes \Hi_2)\longrightarrow \mathcal{B}(\Hi_2),
\end{equation}
such that, for any $X_{1,2}\in \mathcal{B}(\Hi_{1,2})$ and any $X_2 \in \mathcal{B}(\Hi_2)$, it holds that:
\begin{equation}
\label{ }
\trace[\trace_{\Hi_1}[X_{1,2}]X_2]=\trace[X_{1,2}(\mathcal{I}_1\otimes X_2)].
\end{equation}
If now $\{| \psi_k\rangle\}_{k=1}^{d_1}$ and $\{|\phi_l \rangle\}_{l=1}^{d_2}$ are orthonormal bases for $\Hi_1$ and $\Hi_2$ respectively, the partial trace over $\Hi_1$ can be written as:
\begin{equation}
\label{ }
\trace_{\Hi_1}[X_{1,2}]=\sum_{k,l,i}\langle \psi_k \otimes \phi_l | X_{1,2}|\psi_k \otimes \phi_i \rangle |\phi_l\rangle \langle \phi_i|.
\end{equation}
The partial trace over $\Hi_2$ writes in a similar fashion.

The tensor product structure together with the superposition principle enrich quantum theory with the phenomenon of {\em entanglement}. Namely, if some vectors in $|\xi \rangle \in \Hi_{1,2}$ -- like the above-defined basis vectors -- can be {\em factorized} as 
\begin{equation}
\label{eq:entangled}
|\xi\rangle=|\psi\rangle \otimes |\phi\rangle
\end{equation}
for some $|\psi\rangle \in \Hi_1$ and $|\phi \rangle \in \Hi_2$, there exist many vectors in $\Hi_{1,2}$ which cannot be written like \eqref{eq:entangled}; these are called {\em entangled}. 

In the density operator formalism, a state $\rho\in \mathfrak{D}(\Hi_{1,2})$ is called entangled if it cannot be written as a convex combination of factorized operators, i.e.:
\begin{equation}
\label{ }
\rho\neq \sum_{i} \, p_i \, \rho^1_i\otimes \rho^2_i\,\,\,\,\,\,\,\,\,\mbox{with}\,\,\,\,\,\,\,\rho^1_i \in \mathfrak{D}(\Hi_1)\,\,\,\,\mbox{and}\,\,\,\,\rho^2_i \in \mathfrak{D}(\Hi_2) \, .
\end{equation}
Intuitively, this means that an entangled state contains some specifically quantum correlation so that it cannot be separated into subsystem states (conditioned by classical probability correlations). The only individual characterization of a subsystem that we can give is its ``expected'' density operator if information about all other subsystems is ignored. This reduced density operator e.g.~for subsystem 1 is obtained by taking the partial trace over $\Hi_2$ of the overall state $\rho$:
$$\bar{\rho}_1 = \trace_{\Hi_2}[\rho] \, .$$

\subsection{Distinguishing SSC from RSC by measurements}\label{app2}
Accumulating statistics about measurement outcomes at each subsystem separately, allows in principle to detect a pure state for which SSC and RSC are equivalent (see Proposition \ref{prop-meas}). Indeed, the knowledge of all local measurement statistics is equivalent to knowing the reduced state.
\begin{prop}
Except for the case of Proposition \ref{prop-meas}, SSC can only be distinguished from RSC by inspecting correlations between measurement outcomes at different subsystems.
\end{prop}
\proof
The statement builds on the standard fact that the statistics of a local observable $\sigma_1 \otimes \sigma_2 \otimes ... \otimes \sigma_m$ only depend on reduced states $\bar\rho_1,\bar\rho_2,...,\bar\rho_m$. So repeated local measurements can, at their best, fully characterize the $\bar\rho_k$. Checking RSC, i.e.~that these $\bar\rho_k$ are all equal, is thus straightforward. On the other hand, reduced states $\bar\rho_k$ are the best that can be extracted by local measurements in trying to distinguish RSC from SSC states. If $\bar{\rho}_1 = \bar{\rho}_2 = ... =: \bar{\rho}$ have rank one, we have the special case that is always SSC. If instead $\bar{\rho}$ has rank at least 2, we can write it as $\bar{\rho} = p_1 R_1 + p_2 R_2$ where $R_1,\, R_2  \in{\mathfrak B}(\Hi)$, $p_1,\,p_2$ are positive scalars, $R_2$ is positive semidefinite, and $R_1$ is a projector on a 2-dimensional subspace $\mathcal{V}_2$. Consider $R_1 = \q{e_1}\qd{e_1} + \q{e_2}\qd{e_2} = \q{f_1}\qd{f_1} + \q{f_2}\qd{f_2}$, where $\q{e_1},\q{e_2}$ and $\q{f_1},\q{f_2}$ are two orthonormal bases for $\mathcal{V}_2$ with $\langle e_1 \vert f_1 \rangle \notin \{0,1\}$. Now the reconstructed $\bar{\rho}$ could equally well reflect the state 
$$\rho = \bar{\rho} ^{\otimes m}\, ,$$
which is SSC, or e.g.~a state of the form:
\begin{align*}
	&\rho  =  p_2 \, R_2 ^{\otimes m} \\& +\;p_1 (\q{e_1}\q{f_1}+\q{e_2}\q{f_2})(\q{e_1}\q{f_1}+\q{e_2}\q{f_2})^\dag \otimes R_1 ^{\otimes (m-2)} ,
\end{align*}
where the first two subsystems are entangled. This state is not SSC, even for $m=2$.
Thus the local knowledge of $\bar{\rho}$ does not allow to distinguish if the state is SSC or not.
\qed

\subsection{Proof of Proposition \ref{prop:Nogo}}\label{nosmc}

The definition of $\sigma$SMC involves $\trace(\Pi_j^{(k)}\Pi_j^{(\ell)} \rho)$, which takes the partial trace over the state of all subsystems except the pair $\{k,\ell\}$. So we can effectively discard all but two subsystems, and show without loss of generality that it is impossible to make $\sigma$SMC hold for all $\sigma$ on two subsystems $k=1,\, \ell=2$.
In Proposition \ref{propSYM}, we say that $\sigma$SMC for a particular $\sigma$ requires $\Pi_{\rm sym} \rho = \rho$ with $\Pi_{\rm sym}=\sum_j \Pi_j^{\;\otimes m},$ and $\{\Pi_j\}$ the spectral projectors associated to $\sigma$. So if $\sigma$SMC has to hold for both $\sigma$ and $\sigma'$, we must have in particular
$$\Pi_{\rm sym} \Pi_{\rm sym}' \Pi_{\rm sym} \rho = \rho \, ,$$
where $\Pi_{\rm sym}'$ is associated to $\sigma'$. Since $H:= \Pi_{\rm sym}\Pi_{\rm sym}' \Pi_{\rm sym}$ and $\rho$ both are self-adjoint positive semidefinite, the only way to have $H \rho = \rho \neq 0$ is if $H$ has at least one eigenvalue $\geq 1$. Now take in particular $\sigma = \sum_k \; k\,\q{x_k}\qd{x_k}$ and $\sigma' = \sum_k \; k\, \q{p_k}\qd{p_k}$, with $p_k = \frac{1}{\sqrt{n}} \, \sum_{j=0}^{n-1}\, e^{j k 2\pi i/n} \q{x_j}$ (thus the $\q{p_k}$-basis is related to the $\q{x_k}$-basis by Fourier transform). A few computations show that $H$ then has all eigenvalues $<1$, except for $n=2$ that is the case of two qbits. For the latter particular case, one can prove the property by showing e.g.~that there is no state which would satisfy $\sigma$SMC for all $\sigma \in \{\sigma_x,\sigma_y,\sigma_z\}$.
\qed

\bibliographystyle{IEEEtran}
\bibliography{bibQcons2}

\begin{thebibliography}{10}
\providecommand{\url}[1]{#1}
\csname url@rmstyle\endcsname
\providecommand{\newblock}{\relax}
\providecommand{\bibinfo}[2]{#2}
\providecommand\BIBentrySTDinterwordspacing{\spaceskip=0pt\relax}
\providecommand\BIBentryALTinterwordstretchfactor{4}
\providecommand\BIBentryALTinterwordspacing{\spaceskip=\fontdimen2\font plus
\BIBentryALTinterwordstretchfactor\fontdimen3\font minus
  \fontdimen4\font\relax}
\providecommand\BIBforeignlanguage[2]{{%
\expandafter\ifx\csname l@#1\endcsname\relax
\typeout{** WARNING: IEEEtran.bst: No hyphenation pattern has been}%
\typeout{** loaded for the language `#1'. Using the pattern for}%
\typeout{** the default language instead.}%
\else
\language=\csname l@#1\endcsname
\fi
#2}}

\bibitem{Gorinevsky2008}
D.~Gorinevsky, S.~Boyd, and G.~Stein, ``Design of low-bandwidth spatially
  distributed feedback,'' \emph{IEEE Trans.Aut.Cont.}, vol.~53, no.~1, pp.
  257--272, 2008.

\bibitem{Bamieh2002a}
B.~Bamieh, F.~Paganini, and M.~Dahleh, ``Distributed control of spatially
  invariant systems,'' \emph{IEEE Trans.Aut.Cont.}, vol.~47, no.~7, pp.
  1091--1107, 2002.

\bibitem{D'Andrea2003}
R.~D'Andrea and G.~Dullerud, ``Distributed control design for spatially
  interconnected systems,'' \emph{IEEE Trans.Aut.Cont.}, vol.~48, no.~9, pp.
  1478--1495, 2003.

\bibitem{BulloBook}
F.~Bullo, J.~Cort\'es, and S.~Mart{\'\i}nez, \emph{Distributed control of
  robotic networks}.\hskip 1em plus 0.5em minus 0.4em\relax Princeton
  University Press, 2009.

\bibitem{BarHesp1}
P.~Barooah and J.~Hespanha, ``Estimation on graphs from relative measurements:
  distributed algorithms and fundamental limits,'' \emph{IEEE Control Systems
  Magazine}, vol.~27, no.~4, pp. 57--74, 2007.

\bibitem{VaraiyaEst82}
V.~Borkar and P.~Varaiya, ``Asymptotic agreement in distributed estimation,''
  \emph{IEEE Trans. Automatic Control}, vol.~27, no.~3, pp. 650--655, 1982.

\bibitem{BoydADMM}
S.~Boyd, N.~Parikh, E.~Chu, B.~Peleato, and J.~Eckstein, ``Distributed
  optimization and statistical learning via the alternating direction method of
  multipliers,'' \emph{Foundations and Trends in Machine Learning}, vol.~3,
  no.~1, pp. 1--122, 2011.

\bibitem{leonard2007collective}
N.~Leonard, D.~Paley, F.~Lekien, R.~Sepulchre, D.~Fratantoni, and R.~Davis,
  ``Collective motion, sensor networks, and ocean sampling,'' \emph{Proceedings
  of the IEEE}, vol.~95, no.~1, pp. 48--74, 2007.

\bibitem{ZampieriGen}
A.~Chiuso, F.~Fagnani, L.~Schenato, and S.~Zampieri, ``Gossip algorithms for
  simultaneous distributed estimation and classification in sensor networks,''
  \emph{IEEE Journal of Selected Topics in Signal Processing}, vol.~5, no.~4,
  pp. 691--706, 2011.

\bibitem{TsitsiklisThesis}
J.~Tsitsiklis and M.~A. (advisor), ``Problems in decentralized decision making
  and computation,'' \emph{PhD Thesis, MIT}, 1984.

\bibitem{Jadbabaie}
A.~Jadbabaie, J.~Lin, and A.~Morse, ``Coordination of groups of mobile
  autonomous agents using nearest neighbor rules,'' \emph{IEEE Trans. Automatic
  Control}, vol.~48, no.~6, pp. 988--1001, 2003.

\bibitem{ConsensusReview}
R.~Olfati-Saber, J.~Fax, and R.~Murray, ``Consensus and cooperation in
  networked multi-agent systems,'' \emph{Proc. IEEE}, vol.~95, no.~1, pp.
  215--233, 2007.

\bibitem{Moreau2005}
L.~Moreau, ``Stability of multi-agent systems with time-dependent communication
  links,'' \emph{IEEE Trans. Automatic Control}, vol.~50, no.~2, pp. 169--182,
  2005.

\bibitem{nielsen-chuang}
M.~A. Nielsen and I.~L. Chuang, \emph{Quantum Computation and
  Information}.\hskip 1em plus 0.5em minus 0.4em\relax Cambridge University
  Press, Cambridge, 2002.

\bibitem{kitaev-book}
M.~N.~V. A.~Yu.~Kitaev, A. H.~Shen, \emph{Classical and Quantum
  Computation}.\hskip 1em plus 0.5em minus 0.4em\relax American Mathematical
  Society, 2002.

\bibitem{verstraete2009}
F.~Verstraete, M.~M. Wolf, and J.~I. Cirac, ``Quantum computation and
  quantum-state engineering driven by dissipation,'' \emph{Nature Physics},
  vol.~5, pp. 633 -- 636, 2009.

\bibitem{kempe-randomwalks}
J.~Kempe, ``Quantum random walks: An introductory overview,''
  \emph{Contemporary Physics}, vol.~44, no.~4, pp. 307--327, 2003.

\bibitem{szegedy-randomwalks}
M.~Szegedy, ``Quantum speed-up of {M}arkov chain based algorithms,''
  \emph{Proc. 45th Annual IEEE Symp. Foundations of Computer Science}, pp.
  32--41, 2004.

\bibitem{viola-IPSlong}
R.~Blume-Kohout, H.~K. Ng, D.~Poulin, and L.~Viola, ``Information preserving
  structures: A general framework for quantum zero-error information,''
  \emph{preprint}, vol. arxiv, p. 1006.1358v1, 2010.

\bibitem{ticozzi-isometries}
F.~Ticozzi and L.~Viola, ``Quantum information encooding, protection and
  correction via trace-norm isometries,'' \emph{Phys. Rev. A}, vol.~81, no.~3,
  p. 032313, 2010.

\bibitem{ticozzi-QDS}
------, ``Quantum {M}arkovian subsystems: Invariance, attractivity and
  control,'' \emph{IEEE Trans. Aut. Contr.}, vol.~53, no.~9, pp. 2048--2063,
  2008.

\bibitem{ahn-feedback}
C.~Ahn, H.~M. Wiseman, and G.~J. Milburn, ``Quantum error correction for
  continuously detected errors,'' \emph{Phys. Rev. A}, vol.~67, no.~5, pp.
  052\,310:1--11, 2003.

\bibitem{ticozzi-feedbackDD}
F.~Ticozzi and L.~Viola, ``Single-bit feedback and quantum dynamical
  decoupling,'' \emph{Phys. Rev. A}, vol.~74, no.~5, pp. 052\,328:1--11, 2006.

\bibitem{kraus-entanglement}
B.~Kraus, S.~Diehl, A.~Micheli, A.~Kantian, H.~P. B\"{u}chler, and P.~Zoller,
  ``Preparation of entangled states by dissipative quantum markov processes,''
  \emph{Phys. Rev. A}, vol.~78, no.~4, p. 042307, 2008.

\bibitem{ticozzi-QL}
F.~Ticozzi and L.~Viola, ``Stabilizing entangled states with quasi-local
  quantum dynamical semigroups,'' \emph{Phil. Trans. R. Soc. A}, vol. 370, no.
  1979, pp. 5259--5269, 2012.

\bibitem{wolf-gadgets}
M.~Kastoryano, M.~Wolf, and J.~Eisert, ``Precisely timing dissipative quantum
  information processing,'' \emph{arXiv:1205.0985}, 2012.

\bibitem{altafini-tutorial}
C.~Altafini and F.~Ticozzi, ``Modeling and control of quantum systems: An
  introduction,'' \emph{IEEE Trans. Aut. Cont,}, vol.~57, no.~8, pp. 1898
  --1917, 2012.

\bibitem{dalessandro-book}
D.~D'Alessandro, \emph{Introduction to Quantum Control and Dynamics}, ser.
  Applied Mathematics \& Nonlinear Science.\hskip 1em plus 0.5em minus
  0.4em\relax Chapman \& Hall/CRC, 2007.

\bibitem{james-passivity}
M.~James and J.~Gough, ``Quantum dissipative systems and feedback control
  design by interconnection,'' \emph{IEEE Transactions on Automatic Control},
  vol.~55, no.~8, pp. 1806 --1821, 2010.

\bibitem{gough-product}
J.~Gough and M.~R. James, ``The series product and its application to quantum
  feedforward and feedback networks,'' \emph{IEEE Trans. Aut. Contr}, vol.~54,
  no.~11, pp. 2530 --2544, 2009.

\bibitem{bolognani-arxiv}
S.~Bolognani and F.~Ticozzi, ``Engineering stable discrete-time quantum
  dynamics via a canonical {QR} decomposition,'' \emph{IEEE Trans. Aut.
  Contr.}, vol.~55, no.~12, pp. 2721 --2734, 2010.

\bibitem{werner-states}
R.~F. Werner, ``Quantum states with einstein-podolsky-rosen correlations
  admitting a hidden-variable model,'' \emph{Phys. Rev. A}, vol.~40, pp.
  4277--4281, 1989.

\bibitem{werner-symm}
T.~Eggeling and R.~F. Werner, ``Separability properties of tripartite states
  with $u\bigotimes{}u\bigotimes{}u$ symmetry,'' \emph{Phys. Rev. A}, vol.~63,
  p. 042111, 2001.

\bibitem{SSRconsensus}
R.~Sepulchre, A.~Sarlette, and P.~Rouchon, ``Consensus in non-commutative
  spaces,'' \emph{Proc. 49th IEEE Conf. Decision \& Control}, pp. 6596--6601,
  2010.

\bibitem{WildeQIT}
M.~M. Wilde, \emph{From classical to quantum {S}hannon theory}.\hskip 1em plus
  0.5em minus 0.4em\relax McGill University: in preparation, 2011.

\bibitem{sakurai}
J.~J. Sakurai, \emph{Modern Quantum Mechanics}.\hskip 1em plus 0.5em minus
  0.4em\relax Addison-Wesley, New York, 1994.

\bibitem{kraus}
K.~Kraus, \emph{States, Effects, and Operations: Fundamental Notions of
  Quantaum Theory}, ser. Lecture notes in Physics.\hskip 1em plus 0.5em minus
  0.4em\relax Springer-Verlag, Berlin, 1983.

\bibitem{wolf-unital}
C.~B. Mendl and M.~M. Wolf, ``Unital quantum channels - convex structure and
  revivals of birkhoff's theorem,'' \emph{Commun. Math. Phys.}, vol. 289, pp.
  1057--1096, 2009.

\bibitem{albertini-feedback}
F.~Albertini and F.~Ticozzi, ``Discrete-time controllability for feedback
  quantum dynamics,'' \emph{Automatica}, vol.~47, no.~11, pp. 2451 -- 2456,
  2011.

\bibitem{viola-engineering}
S.~Lloyd and L.~Viola, ``Engineering quantum dynamics,'' \emph{Phys. Rev. A},
  vol.~65, pp. 010\,101:1--4, 2001.

\bibitem{BoydGossip}
S.~Boyd, A.~Ghosh, B.~Prabhakar, and D.~Shah, ``Randomized gossip algorithms,''
  \emph{IEEE Trans. Information Theory (Special issue)}, vol.~52, no.~6, pp.
  2508--2530, 2006.

\bibitem{alicki-lendi}
R.~Alicki and K.~Lendi, \emph{Quantum Dynamical Semigroups and
  Applications}.\hskip 1em plus 0.5em minus 0.4em\relax Springer-Verlag,
  Berlin, 1987.

\end{thebibliography}

\end{document}